\documentclass[preprint]{elsarticle} % options: preprint or review

\makeatletter
	\def\ps@pprintTitle{%
 	\let\@oddhead\@empty
	\let\@evenhead\@empty
	\def\@oddfoot{\centerline{\thepage}}%
	\let\@evenfoot\@oddfoot}
\makeatother

\usepackage{etoolbox}
\patchcmd{\MaketitleBox}{\footnotesize\itshape\elsaddress\par\vskip36pt}{\footnotesize\itshape\elsaddress\par\parbox[b][36pt]{\linewidth}{\vfill\hfill\textnormal{\today}\hfill\null\vfill}}{}{}%
\patchcmd{\pprintMaketitle}{\footnotesize\itshape\elsaddress\par\vskip36pt}{\footnotesize\itshape\elsaddress\par\parbox[b][36pt]{\linewidth}{\vfill\hfill\textnormal{\today}\hfill\null\vfill}}{}{}%

\usepackage[%pdfmark,%
            colorlinks=true,%
            breaklinks=true,%
            linkcolor=blue,
            urlcolor=blue,%
            citecolor=blue,%
            pdftitle={Title of paper}, 
            pdfkeywords={Keywords},	
            pdfauthor={Authors},		
            bookmarksopen=false,
            pdfpagemode=UseNone]{hyperref}
            
\usepackage[margin = 1.2in]{geometry}
\usepackage[T1]{fontenc}
\usepackage[english]{babel}
\usepackage[utf8]{inputenc}
\usepackage{mathtools}
\usepackage{amsmath}
\usepackage{amsfonts}
\usepackage{amsthm}
\usepackage{amssymb}
\usepackage{array}
\usepackage{algorithm} %ctan.org\pkg\algorithms
\usepackage{algorithmicx}
\usepackage{algpseudocode}
\usepackage{subcaption}
\usepackage{siunitx}
\usepackage{arydshln}
\usepackage{nomencl}
\usepackage{framed}  % For creating a box around the nomenclature
\usepackage{multicol}  % For creating two columns

\usepackage{mathrsfs}
\usepackage{upgreek}

\usepackage{color}
\usepackage{xcolor}
\usepackage{url}
\usepackage{cleveref}
\usepackage{graphicx}
\usepackage{breakcites}
\usepackage{soul} % for using the \ul command with linebreak
\usepackage{enumitem}
\usepackage[title,titletoc,toc]{appendix}
\usepackage{tikz}
\usetikzlibrary{arrows.meta, positioning}
\usepackage{pgfplots}
\pgfplotsset{compat=1.16}

\usepackage{tabularx}

\usepackage[textsize=tiny]{todonotes}

\definecolor{darkgreen}{rgb}{0.0, 0.5, 0.0}

\newtheorem{assumption}{Assumption}
\newtheorem{proposition}{Proposition}

\renewenvironment{proof}%
{\noindent {\textbf{Proof}:} }%
{\hfill $\Box$ \\[1ex] }

\newcommand{\bit}{\begin{itemize}}
\newcommand{\eit}{\end{itemize}}
\newcommand{\ben}{\begin{enumerate}}
\newcommand{\een}{\end{enumerate}}

\newcommand {\real} {\mathbb{R}}

\DeclareMathOperator*{\Cov}{\mathbb{C}\text{ov}}
\newcommand{\rd}{\text{\upshape d}} % Used for "d" in integration "dx".)

\newcommand{\bA}{\ensuremath{\mathbf{A}}}
\newcommand{\bB}{\ensuremath{\mathbf{B}}}
\newcommand{\bC}{\ensuremath{\mathbf{C}}}
\newcommand{\bD}{\ensuremath{\mathbf{D}}}
\newcommand{\bE}{\ensuremath{\mathbf{E}}}

\newcommand{\bH}{\ensuremath{\mathbf{H}}}
\newcommand{\bI}{\ensuremath{\mathbf{I}}}

\newcommand{\bK}{\ensuremath{\mathbf{K}}}

\newcommand{\bM}{\ensuremath{\mathbf{M}}}
\newcommand{\bN}{\ensuremath{\mathbf{N}}}
\newcommand{\bO}{\ensuremath{\mathbf{O}}}
\newcommand{\bP}{\ensuremath{\mathbf{P}}}
\newcommand{\bQ}{\ensuremath{\mathbf{Q}}}
\newcommand{\bR}{\ensuremath{\mathbf{R}}}

\newcommand{\bU}{\ensuremath{\mathbf{U}}}

\newcommand{\bW}{\ensuremath{\mathbf{W}}}
\newcommand{\bX}{\ensuremath{\mathbf{X}}}

\newcommand{\ba}{\ensuremath{\mathbf{a}}}

\newcommand{\bu}{\ensuremath{\mathbf{u}}}

\newcommand {\bgamma} {\mbox{\boldmath $\gamma$}}

\newcommand {\bPhi} {\mbox{\boldmath $\Phi$}}

\newcommand {\bxi} {\mbox{\boldmath $\xi$}}

\newcommand{\cE}{\ensuremath{\mathcal{E}}}
\newcommand{\cF}{\ensuremath{\mathcal{F}}}

\newcommand{\cN}{\ensuremath{\mathcal{N}}}

\newcommand{\cT}{\ensuremath{\mathcal{T}}}

\graphicspath{{./figures/}}

\begin{document}
	
    \begin{frontmatter}
        
        \title{Stochastic Operator Inference for reduced-order modeling of capillary wave turbulence using experimental measurements
        }
        %\tnotetext[mytitlenote]{}
        
        % \author[affil1]{Author A}
        % \author[affil2]{Author B}
        
        \author[affil1]{Hyeonghun Kim\corref{cor1}}
        \ead{hyk049@ucsd.edu}
        \author[affil1]{Lei Zhang}
        \author[affil2]{James Friend}
        \author[affil1]{Boris Kramer}
        
        \cortext[cor1]{Corresponding author}
        
        % \address[affil1] {Research Center of the Flat Earth Society}
        % \address[affil2] {Elliptic Science, Inc.}
        
        \address[affil1]{Department of Mechanical and Aerospace Engineering, University of California San Diego, CA, United States}
        \address[affil2]{Mechanical Engineering and Materials Science, Washington University in St. Louis, MO, United States}
        
        \begin{abstract}
        Modeling complex physical phenomena directly from experimental data poses fundamental challenges: measurement devices introduce noise and biases into the data, the governing equations are often unknown or intractable, and the dynamics observed experimentally may exhibit stochastic behavior.
        In this paper, we consider capillary wave turbulence---an example of nonlinear wave interactions at a microfluidic interface---measured by ultra-high-speed digital holographic microscopy.
        With the goal to learn an efficient model directly from these data, we adapt a stochastic extension of Operator Inference to learn low-dimensional stochastic differential equation representations of the microscale wave dynamics from experimental measurements. Each experiment is repeated 16 times, and 10 different conditions are created by changing the nondimensional acoustic capillary number through an excitation device.
        In addition, we propose a new strategy to select the reduced-order model dimension, based on both mean and covariance errors. We further add Tikhonov regularization to the stochastic Operator Inference framework and show that, beyond its conventional role as a numerical stabilization technique in deterministic settings, it has a physically meaningful interpretation as controlling the spectral content of the learned stochastic dynamics.
        We demonstrate that the resulting stochastic reduced-order models faithfully capture salient physical features of capillary wave turbulence across a range of experimental conditions.
        \end{abstract}	
        
		
    \end{frontmatter}

% Define custom colors to match Matplotlib
\definecolor{C00}{rgb}{0.121, 0.467, 0.705} % skyblue
\definecolor{C01}{rgb}{1.000, 0.498, 0.055} % Orange
\definecolor{C02}{rgb}{0.172, 0.627, 0.173} % Green
\definecolor{C03}{rgb}{0.839, 0.153, 0.157} % red
\definecolor{C04}{rgb}{0.580, 0.404, 0.741} % purple
\definecolor{C05}{rgb}{0.549, 0.337, 0.294} % brown
\definecolor{C06}{rgb}{0.890, 0.467, 0.761} % pink
\definecolor{C07}{rgb}{0.498, 0.498, 0.498} % gray
\definecolor{C08}{rgb}{0.737, 0.741, 0.133} % yellow
\definecolor{C09}{rgb}{0.090, 0.745, 0.812} % cyan
\definecolor{C10}{rgb}{0.121, 0.467, 0.705} % skyblue

%%%%%%%%%%%%%%%%%%%%%%%%%%%%%%%%%%%%%%%%
%%%%%%%%%%%%%%%%%%%%%%%%%%%%%%%%%%%%%%%
\section{Introduction} \label{sec:intro}
%%%%%%%%%%%%%%%%%%%%%%%%%%%%%%%%%%%%%%%

% 1. What is the problem?
Advances in high-fidelity numerical simulations and experimental techniques have enabled increasingly sophisticated characterization and understanding of complex physical phenomena.
Despite these advances, obtaining such high-fidelity descriptions remains computationally and experimentally expensive.
High-fidelity numerical simulations often require solving high-dimensional systems of semi-discretized partial differential equations (PDEs), while experimental campaigns demand substantial time, resources, and careful control of operating conditions.
This often poses a big challenge for real-time prediction and outer-loop tasks, including uncertainty quantification, design optimization,~etc.
\par
Reduced-order models (ROMs) address this challenge by constructing low-dimensional surrogate models that enable computationally efficient reconstruction and prediction of high-dimensional system states, while preserving their dominant dynamics; see, e.g.,~\cite{benner2015survey, quarteroni2014reduced, brunton2022data, Kramer2026digitaltwin}.
In particular, frameworks that learn ROMs from data have provided a promising solution in engineering and science~\cite{ghattas2021learning}.
Over the past decade, several ROM learning methods---including Dynamic Mode Decomposition (DMD)~\cite{schmid2010dynamic}, Sparse Identification of Nonlinear Dynamics (SINDy)~\cite{brunton2016discovering}, and Operator Inference (OpInf)~\cite{PEHERSTORFER2016196, kramer2024learning}---have been successfully applied across many science and engineering problems.
In parallel, deep learning approaches based on neural ordinary differential equations~\cite{chen2018neural} have also been explored to learn ROMs directly from data~\cite{caldana2025neural}, offering a more expressive but less interpretable alternative to these operator-based methods.
The advantages of ROM learning methods become even more powerful in experimental settings, since the experimental systems often do not admit a tractable mathematical model that faithfully reproduces the observed dynamics. Therefore, the corresponding ROMs must be learned directly from measurement data.
\par
%
% 2. Which limitations remain? -- stochasticity in experimental systems
Most ROM approaches have been developed under the assumption that the underlying dynamics are deterministic, with randomness entering only as additive measurement or simulation noise, or uncertain model parameters, in each case generated by random number generators with mostly white noise assumptions.
This assumption is often violated in real experimental systems, where the observed dynamics exhibit intrinsic stochasticity arising from unresolved physical processes, chaotic dynamics, and experimental variability even under carefully controlled environments.
In such cases, randomness is not merely a measurement artifact but an intrinsic feature of the system dynamics. 
\par
%
% 3. What has been done?
Existing efforts have addressed uncertainty in deterministic ROM learning in several approaches. Bayesian formulations of Operator Inference~\cite{guo2022bayesian, MCQUARRIE2025134572}, for example, treat ROM operator learning as a Bayesian inference problem, where the ROM operators are treated as random variables and inferred probabilistically from the observed state data, thereby improving robustness to measurement noise and parameter uncertainty.
However, these methods still assume the underlying dynamics to be deterministic, with uncertainty arising from limited observations or uncertain parameters---not from any inherent stochasticity in the system.
Similarly, the ensemble-SINDy method~\cite{fasel2022ensemble} enables robust model discovery in a noisy and limited-data setting.
Although the method has shown promising results for high-dimensional simulated systems, its application to real-world data has been limited to low-dimensional datasets, such as those approximated by a two-state Lotka--Volterra system.
%
% Data-driven ROMs for experimental systems
In parallel, several ROM learning approaches have been applied to complex experimental systems. For example, \cite{raut2025arousal} used time-delay embedding and Koopman-based methods to model large-scale brain physiology dynamics as a latent dynamical system. 
The authors in~\cite{tomasetto2025reduced} combined an LSTM-based temporal encoder with a shallow decoder network to learn ROMs from sparse sensor measurements, demonstrating strong performance on both simulated and experimental fluid flows.
Within the OpInf framework, a recent work~\cite{peterson2026data} applied the deterministic OpInf framework to experimental data from a pilot-scale methanation reactor.
Despite demonstrating the feasibility of learning ROMs from experimental data, these approaches do not yet explicitly model stochastic latent dynamics to account for the inherent variability present in many experimental systems.
\par
%
% Stochastic approaches
In contrast to these deterministic approaches, several methods have been developed to learn stochastic ROMs.
The authors in~\cite{boninsegna2018sparse} extend SINDy to discover stochastic differential equations (SDEs) latent dynamics models.
Another line of work~\cite{galioto2020bayesian} formulates system identification as probabilistic inference through a hidden Markov model governing discrete-time stochastic dynamics, enabling simultaneous treatment of parameter, model-form, and measurement uncertainty.
The authors in~\cite{Dietrich2021Learning} introduce the approach of identifying the system of SDEs and approximating its drift and diffusion terms through neural networks, by incorporating the stochastic numerical integrators into the training loss functions.
While these methods provide promising directions for learning stochastic systems, their validation has largely focused on relatively low-dimensional synthetic systems. Examples include one-dimensional potentials and projected two-dimensional diffusion processes~\cite{boninsegna2018sparse}; canonical dynamical systems, such as the pendulum, Van der Pol oscillator, Lorenz ’63, and reaction--diffusion equation, all with fewer than three state variables~\cite{galioto2020bayesian}; and other low-dimensional stochastic systems, including a 19-dimensional toy problem~\cite{Dietrich2021Learning}. 
More recently, the authors in~\cite{freitag2025learning} introduced a stochastic extension of the OpInf ROM framework that overcomes this dimensionality limitation by learning SDE-based ROM directly from high-dimensional data. 
However, its validation has thus far been restricted to synthetic high-dimensional data generated by discretizing one- and two-dimensional linear stochastic heat equations, where the governing SDEs are known a priori and tightly controlled noise is inserted at every time step. 
Thus, whether stochastic OpInf ROM learning frameworks generalize to real high-dimensional experimental systems, where governing equations are unknown and the stochasticity is intrinsic, remains unexplored.
\par
In this paper, we address these gaps by extending the stochastic OpInf ROM framework to a previously unexplored setting: high-dimensional experimental data.  
Specifically, we demonstrate the framework on digital holographic microscopy (DHM) measurements of capillary wave turbulence~\cite{orosco2023identification, connacher2024direct, zhang2023onset, blamey2013microscale}, for which neither the governing equations nor a first-principles numerical discretization are available.
%
% Q. Why capillary wave turbulence?
Capillary wave turbulence is a challenging case for this model learning.
Its theory was built on stationary kinetic equations by the pioneering works ~\cite{zakharov1967weak, pushkarev1996turbulence, pushkarev2000turbulence}, which provide a statistical spectral solution for the wave energy distribution.
Since then, the phenomenon has been investigated extensively, both numerically~\cite{deike2014direct, pan2017understanding} and experimentally~\cite{falcon2007observation, falcon2009capillary, deike2012decay}, also see the surveys~\cite{newell2011wave, falcon2022experiments}.
However, these studies validate aggregate spectral quantities, such as the Kolmogorov--Zakharov exponent and energy flux, rather than the realization-specific time evolution of the measured surface height.
Moreover, some aspects of the experimental setup in~\Cref{sec:capillary_wave_turbulence_section} make a first-principles forward model impractical or unidentifiable.
First, the measured surface height is an indirect quantity reconstructed from raw digital holographic images.
Second, the DHM measurements are spatially confined to the small laser-passing region rather than the entire fluid domain, making the effective boundaries of the observed region and the corresponding boundary conditions unknown.
%
% 5. What is the paper about?
\par
The contributions of this work are as follows.
First, we extend the recently introduced stochastic ROM framework based on Operator Inference~\cite{freitag2025learning} from synthetic benchmark problems to the significantly more complex problem of low-dimensional modeling of capillary wave turbulence using experimental DHM measurements. 
This represents the first demonstration of stochastic OpInf on real experimental data. It provides evidence that the framework---together with the extensions we introduce---can learn accurate ROMs directly from high-dimensional measurements, where the governing equations of the underlying real-world stochastic system are unknown.
Second, we augment the stochastic OpInf learning problem from~\cite{freitag2025learning} with a Tikhonov regularization.
While the regularization technique has been widely applied in the deterministic OpInf cases~\cite{McQuarrie_2021, mcquarrie2023nonintrusive, issan2023predicting, SAWANT2023115836, qian2022reduced, KIM2025114418, geelen2023learning, kang2026}, its use in stochastic OpInf is new. 
We show that the regularization hyperparameters carry physical significance in the stochastic ROM setting, beyond their traditional role in improving numerical stability and accuracy.
In particular, we demonstrate that they can be used to (i) alleviate aliasing effects and (ii) to modulate the level of stochastic energy in the ROM.
Third, we propose a new ROM dimension selection strategy based on the tradeoff between two statistical quantities central to stochastic modeling: the weak mean and covariance errors.
This provides an alternative to the conventional approach for selecting the ROM dimension based on snapshot energy retention.
Fourth, we provide a detailed analysis of the covariance evolution of the learned stochastic ROMs under the implicit Euler--Maruyama scheme~\cite{higham2001algorithmic} and reveal the coupled roles of the drift and diffusion operators in governing the learned ROM covariance dynamics. This analysis provides insight into the sources of covariance errors and offers a diagnostic perspective for assessing stochastic ROM accuracy. 
\par
The rest of the paper is organized as follows. \Cref{sec:capillary_wave_turbulence_section} describes the experimental setup and data collection for the capillary wave turbulence. \Cref{sec:stochastic_OpInf_ROM} introduces the stochastic OpInf learning framework, including the role of regularization and the proposed ROM dimension selection method. \Cref{sec:numerical_results} presents the numerical results, and \Cref{sec:conclusion} gives conclusions and an outlook to future~work.

%%%%%%%%%%%%%%%%%%%%%%%% Capillary wave turbulence %%%%%%%%%%%%%%%%%%%%%%%
\section{Experimental capillary wave turbulence measured by digital holographic microscopy}
\label{sec:capillary_wave_turbulence_section}

We first discuss the experimental system used to generate and measure capillary wave turbulence. \Cref{ss:DHM_experimental_setup} describes the digital holographic microscopy (DHM) setup and the experiment details, and \Cref{ss:DHM_experimental_data} presents the resulting measurement datasets.

\subsection{Experimental setup}
\label{ss:DHM_experimental_setup}
Capillary wave turbulence has been characterized experimentally by fluid surface displacement, measured by various techniques, see, e.g.,~\cite[Sec.~6]{falcon2022experiments}.
In our case, surface displacements are measured via digital holographic microscopy (DHM, Lynce\'e Tex SA, Lausanne, Switzerland~\cite{emery2021metrology}), which utilizes holographic imaging methods, coupled with a custom ultra high-speed camera (FASTCAM NOVA S12, Photron, San Diego, CA, United States). Together, they enable high-fidelity holographic imaging of rapidly vibrating microfluidic interfaces.
\Cref{fig:DHM_setup_droplet_shcematic} provides a schematic overview of the DHM system and its main components.
The experimental substrate consists of a transparent, single-crystal lithium niobate piezoelectric transducer with a $127.69^\circ$ Y-rotated cut and a thickness of $500 \, \upmu \textnormal{m}$, designed to operate in the fundamental thickness-mode resonance~\cite{Collingnon2018}. Electrodes are deposited on both sides of the wafer, leaving a $6.35 \, \textnormal{mm}$ diameter transparent window through which the sessile fluid drop is observed. The drop is confined within a $200\,\upmu\textnormal{m}$-thick polyimide annulus with an inner diameter $9{,}525\, \upmu \textnormal{m}$, which enables reproducible experimental conditions across repeated trials.
\begin{figure}
    \centering
    \includegraphics[width=1.00\linewidth]{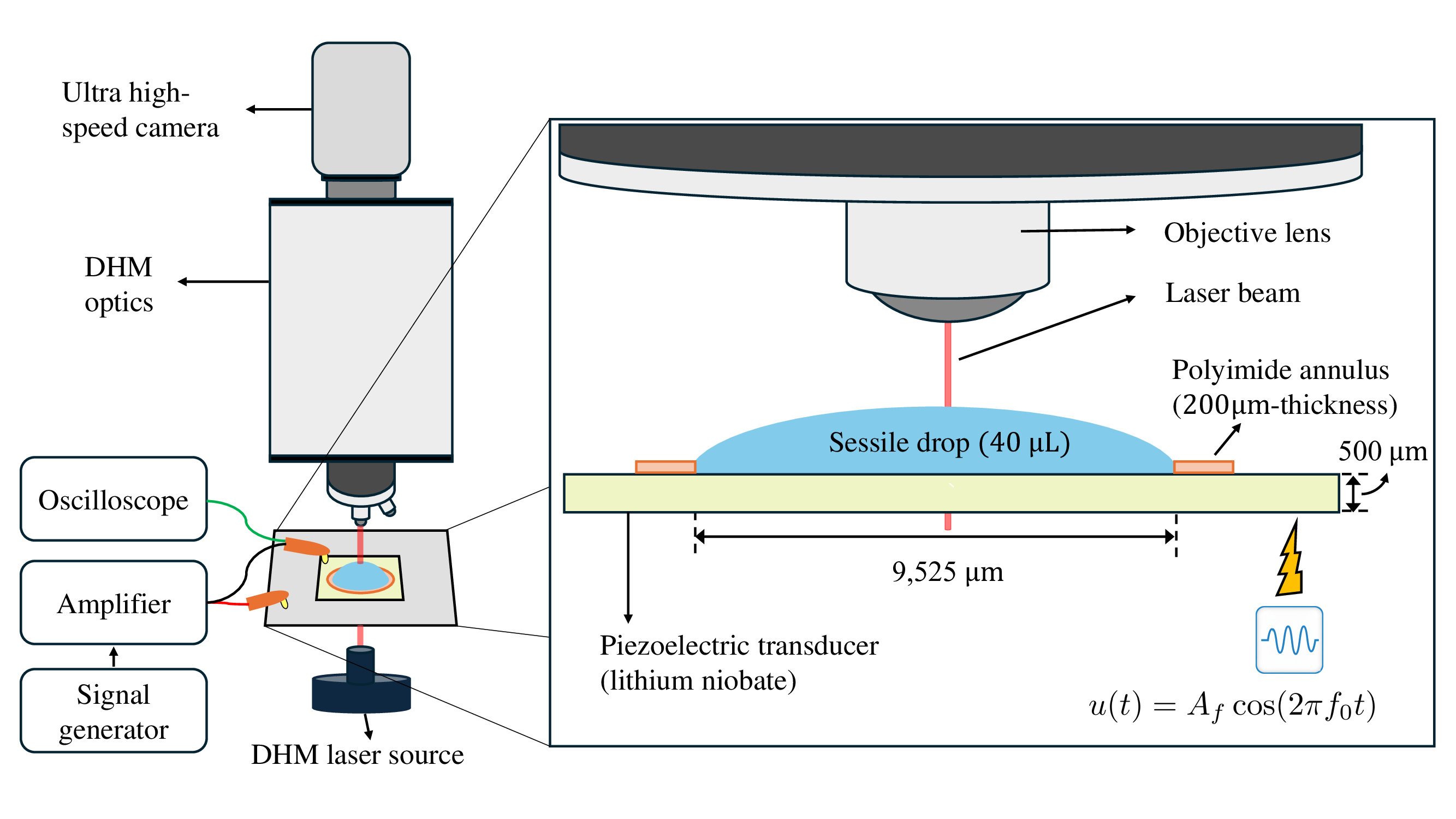}
    \caption{Digital holographic microscopy and schematic of microfluidic experimental setup. A sessile fluid droplet ($40 \, \upmu \textnormal{L}$) sits on a piezoelectric transducer, which is excited by an acoustic signal. The fluid surface displacement is measured holographically around the droplet apex area by the ultra-high-speed camera.}
    \label{fig:DHM_setup_droplet_shcematic}
\end{figure}
The transducer is excited by a sinusoidal acoustic signal at approximately $7 \times 10^6 \,\textnormal{Hz}$, driving its fundamental thickness-mode resonance~\cite{vasan2020fabrication}. The signal is generated and amplified before being applied to the substrate, inducing capillary wave turbulence at the microfluidic interface.
A high-intensity laser beam with a wavelength of $660 \, \textnormal{nm}$ is divided into reference and measurement paths, and their interference pattern is recorded by a camera. Variations in the optical path distance through the fluid cause phase shifts, which are reconstructed into a two-dimensional surface displacement map.
Phase unwrapping is applied when necessary to recover height variations exceeding one wavelength.
The DHM system records the resulting hologram images at a rate of $115{,}200$ frames per second (sampling frequency of 115.2 kHz) with a resolution of $10\,\textnormal{nm}$ in the surface displacement direction and on the order of $1.5\,\upmu\textnormal{m}$ along the lateral image plane.
The raw holograms are converted into spatially resolved wave-field measurement data, which are subsequently used as the state variables in the ROM framework.
For further details on the DHM technology and its experimental details, see~\cite{cuche1999simultaneous} and~\cite[Sect.~2]{orosco2023identification}, respectively.
The forcing amplitude~$A_f$ of the prescribed input signal $u(t)=A_f \cos(2 \pi f_0 t)$ is the only controllable experimental parameter. For each forcing amplitude at the fixed excitation frequency of $f_0 = 7 \times 10^6 \, \textnormal{Hz}$, the experiment is repeated $\textnormal{N}_{\textnormal{rep}}$ times under identical conditions. 
The treatment of $A_f$ in the ROM learning problem is discussed further in~\Cref{sss:mitigate_aliasing}.
For the varying values of forcing amplitude, we use the nondimensional acoustic capillary number, as
\begin{equation}
    \textnormal{Ca}_\textnormal{ac} = \frac{\mu \cdot V(A_f)}{\sigma},
\label{eq:acoustic_capillary_number}
\end{equation}
where $\mu$ is the dynamic viscosity\,[$\textnormal{Pa}\cdot\textnormal{s}$], $V(A_f)$ is the particle velocity at the substrate surface\,[$\mathrm{m}/\textnormal{s}$], and $\sigma$ is the surface tension\,[$\mathrm{N}/\textnormal{m}$].
The particle velocity $V(A_f)$ is the key quantity that depends on the varying amplitudes of the acoustic excitation signal. This is measured as a scalar value directly by the laser Doppler vibrometer~\cite{friend2010using, fukushima2009performance} (LDV, MSA-400 Microscope Vibrometer, Polytec, Waldbronn, Germany).

\subsection{Experimental data}
\label{ss:DHM_experimental_data}
The complete dataset, including all repetitions and acoustic capillary number cases, is publicly available through~\cite{capillary_wave_datasets}. 
All experiments are conducted using $40\,\upmu\textnormal{L}$ of deionized water, with fixed material properties: density $\rho=998 \, \textnormal{kg}/\textnormal{m}^3$ and surface tension $\sigma = 0.0725 \,\textnormal{N}/\textnormal{m}$.
A total of 16 repeated experiments are done for each of 10 operating conditions, which are defined by different values of nondimensional acoustic capillary number~\eqref{eq:acoustic_capillary_number} obtained by varying the excitation signal amplitude. 
The measurements capture the spatiotemporal evolution of the water surface height around the droplet apex area, see~\Cref{fig:DHM_setup_droplet_shcematic}. 
Here, the recorded DHM data represent the relative surface-height displacement from the resting height within the small region probed by the laser beam.
At low excitation powers, the surface-height variations are on the order of hundreds of micrometers, while substantially larger variations, reaching several millimeters, are observed as the forcing amplitude $A_f$ increases.
These variations are well above the displacement resolution of $10\,\textnormal{nm}$ of the DHM system (see~\Cref{ss:DHM_experimental_setup}), the instrument noise floor of approximately $2.85 \, \textnormal{nm}$, and noise associated with the thermal evaporation of the fluid with a rate of approximately $1 \, \textnormal{nm}$ per minute~\cite{yu2025fluid}. This indicates that the measured surface-height fluctuations are well resolved by the measurement system.
For each acoustic capillary number, the data are stored in HDF5 files containing the surface height data together with metadata describing the spatial and temporal discretization. 
Each dataset consists of $100{,}001$ snapshots spanning approximately $T_{f} = 0.868\,\textnormal{s}$, over a $300\,\upmu\textnormal{m} \times 300\,\upmu\textnormal{m}$ two-dimensional measurement interface, discretized into a $200 \times 200$ spatial grid.
Each snapshot is represented by a $200 \times 200$ matrix of real-valued surface height [$\upmu \textnormal{m}$].
The total size of the 16 repeated experimental datasets for each acoustic capillary number case is on the order of $200 \, \textnormal{GB}$, and the complete collection of 10 acoustic capillary number cases is approximately $1.9 \, \textnormal{TB}$.

%%%%%%%%%%%%%%%%%%%%%%%% Operator Inference %%%%%%%%%%%%%%%%%%%%%%%
\section{Learning stochastic reduced-order models via Operator Inference} \label{sec:stochastic_OpInf_ROM}
In this section, we introduce the stochastic ROM learning framework based on Operator Inference. \Cref{ss:dimensionality_reduction_for_stochastic_ROM} describes the dimensionality reduction used in the stochastic ROM learning setting, and \Cref{ss:stochastic_modeling_asumptions_and_choices} discusses the modeling assumptions used for stochastic ROMs of capillary wave turbulence. \Cref{ss:ROM_SDE_mean_covariance_dynamics_derivation} derives the mean and covariance dynamics of the reduced states, providing additional details and making explicit several intermediate steps from the formulation in~\cite{freitag2025learning}. \Cref{ss:SDE_ROMs_learning_for_capillary_wave_system} formulates the stochastic ROM learning problems for learning the drift and diffusion operators of the SDE ROMs. \Cref{ss:regularization_hyperparameters_selection_based_on_spectral_contents} incorporates Tikhonov regularization into the stochastic OpInf ROM learning problems and analyzes its physical significances, including its role in mitigating aliasing effects and modulating stochastic energy in the ROM.
Finally, \Cref{ss:r_selection_based_on_mean_cov_error_tradeoff} proposes a new ROM dimension selection strategy based on two statistical moments---the mean and covariance---that are critical to stochastic ROM learning.
%
%%%%%%%%
\subsection{Dimensionality reduction for stochastic reduced-order model learning}
\label{ss:dimensionality_reduction_for_stochastic_ROM}
Our goal is to model the low-dimensional representation of the stochastic dynamics of capillary wave turbulence using experimental measurements of the fluid surface displacement.
In this study, we use one-dimensional spatial profiles extracted from the DHM measurements, which are acquired on a $200 \times 200$ spatial grid, as outlined in~\Cref{sec:capillary_wave_turbulence_section}.
Because the measurement area is sufficiently small (see~\Cref{fig:DHM_setup_droplet_shcematic}), the 200 spatial profiles along either spatial direction exhibit nearly identical behavior in both the time and frequency domains. We therefore use the centerline of the two-dimensional measurement as a representative one-dimensional spatial profile. This retains the relevant spatial characteristics of the measured dynamics while substantially reducing the computational cost of the subsequent ROM simulations.
Each extracted profile has a spatial dimension of $n=200$.
We then build the state snapshot matrix from profiles collected across repeated experiments under an identical forcing input condition as
\begin{equation}
    \mathbb{X} \coloneqq [ \underbrace{\left[ \mathbf{X}(t_0,\omega_1), \ldots, \mathbf{X}(t_0,\omega_L) \right]}_{\mathbb{X}_0: L \ \textnormal{measurement samples at}\ t=t_0}, \ldots, \underbrace{\left[ \mathbf{X}(t_s,\omega_1), \ldots, \mathbf{X}(t_s,\omega_L) \right]}_{\mathbb{X}_s: L \ \textnormal{measurement samples at}\ t=t_s} ] \in \mathbb{R}^{n \times (s+1)L}.
\label{eq:high_fidelity_state_snapshot}
\end{equation}
Here, $\bX(t_k, \omega_j) \in \real^n$ denotes the measurement states---represented as surface displacements---at discrete time instances $t_k \in [0,T]$, $k=0,\ldots,s$, and for stochastic process observations~$\omega_j \in \Omega$, $j=1,\ldots,L$. Each block matrix~$\mathbb{X}_k$ represents an ensemble of $L$ measurements at the time instance~$t_k$.
The forcing input amplitude $A_f$ is omitted as an explicit input, since each snapshot matrix is constructed under a fixed forcing input condition. 
In a parametric learning setting, however, $A_f$ should be included as an input parameter.
The discrete time instances are determined by the camera sampling rate of $115.2 \, \textnormal{kHz}$, yielding a uniform time step of approximately $8.68 \times 10^{-6}$ seconds, for each experiment of duration $T_f = 0.868 \, \textnormal{s}$.
Nondimensionalization of the state snapshot matrix has been widely used in previous OpInf studies to improve numerical stability, particularly when the state vector consists of significantly different physical scales or orders of magnitude~\cite{Swischuk_2020, McQuarrie_2021}. In this study, however, the state consists of a single variable, wave surface height, whose values do not span multiple orders of magnitude. Therefore, no state scaling or nondimensionalization is applied.
To reduce the dimension of the data, we take the proper orthogonal decomposition (POD) approach, by which we obtain the linear basis of the given snapshots~\eqref{eq:high_fidelity_state_snapshot} via the reduced singular value decomposition (SVD). For $n \leq (s+1)L$, this gives
\begin{equation*}
    \mathbb{X} = \mathbf{\Phi} \mathbf{\Sigma} \mathbf{\Psi}^\top,
% \label{eq:SVD_of_state_snapshot_matrix}
\end{equation*}
where the columns of $\mathbf{\Phi} \in \real^{n \times n}$ and $\mathbf{\Psi} \in \real^{(s+1)L \times n}$ are the left and right singular vectors of~$\mathbb{X}$, respectively, and $\mathbf{\Sigma} \in \real^{n \times n}$ is a diagonal matrix with singular values~$\sigma_1 \geq \sigma_2 \geq \cdots \geq \sigma_n > 0$ on its diagonal.
The leading $r$ columns of the left singular vector matrix~$\mathbf{\Phi}$ form the rank-$r$ linear POD basis matrix $\mathbf{\Phi}_r = [\boldsymbol{\varphi}_1, \ldots, \boldsymbol{\varphi}_r] \in \mathbb{R}^{n \times r}$ whose columns are orthonormal, i.e., $\mathbf{\Phi}_r^\top {\mathbf{\Phi}}_r = \bI_r \in \real^{r \times r}$. The basis~$\mathbf{\Phi}_r$ defines a linear projection $\mathbf{\Phi}_r \mathbf{\Phi}_r^\top \in \mathbb{R}^{n \times n}$ onto the $r$-dimensional subspace spanned by~$\boldsymbol{\varphi}_1,\ldots,\boldsymbol{\varphi}_r$.
Using the POD basis, we obtain the reduced state snapshot matrix via projection of~\eqref{eq:high_fidelity_state_snapshot} as
\begin{equation}
    \widehat{\mathbb{X}} = \mathbf{\Phi}_r^\top \mathbb{X} = [ \underbrace{\left[ \widehat{\bX}(t_0,\omega_1), \ldots, \widehat{\bX}(t_0,\omega_L) \right]}_{\widehat{\mathbb{X}}_{0}: L \ \textnormal{reduced state samples at}\ t=t_0}, \ldots, \underbrace{\left[ \widehat{\bX}(t_s,\omega_1), \ldots, \widehat{\bX}(t_s,\omega_L) \right]}_{\widehat{\mathbb{X}}_{s}: L \ \textnormal{reduced state samples at}\ t=t_s} ] \in \mathbb{R}^{r \times (s+1)L},
\label{eq:reduced_state_snapshot}
\end{equation}
where $\widehat{\bX}(t_k, \omega_j) = \bPhi_r^\top \bX(t_k, \omega_j) \in \real^r$ denotes the reduced state at the  $k$th time step for the $j$th noise realization.
We learn an underlying ROM using the trajectories of the reduced states across all noise realizations.
This motivates the following modeling assumptions required for the identification of the reduced dynamics of the experimental capillary wave dynamics. 

%%%%%%%%%%%%%%%%%%%%%% Non-intrusive learning of low-fidelity model %%%%%%%%%%%%%%%%%%%%%%%%%%%%
\subsection{Assumptions and modeling choices}
\label{ss:stochastic_modeling_asumptions_and_choices}
\begin{assumption} 
The capillary wave dynamics admit a low-dimensional representation where the reduced states evolve through a stochastic differential equation.
\end{assumption}
\noindent We aim to learn bilinear SDE ROMs with additive noise:
\begin{equation}
    \mathrm{d} \widehat{\bX}_{t} = \left[ \widehat{\bA} \widehat{\bX}_{t} + \widehat{\bB} \bu(t) + \sum_{\ell=1}^{m} \widehat{\bN}_{\ell} \widehat{\bX}_{t} u_\ell(t) \right] \mathrm{d} t + \widehat{\bM} \mathrm{d} \bW_t,
\label{eq:bilinear_stochastic_ROM}
\end{equation}
where $\widehat{\bX}_{t} \in \real^r$ denotes the reduced state, $\widehat{\bX}_{0} \in \real^r$ is the initial condition, $\bW_t \in \real^{d}$ denotes the $d$-dimensional Wiener process with correlation matrix $\mathbf{K}\in\mathbb{R}^{d \times d}$, $t \in [0, T]$, and $\bu(t) \in \real^m$ is a deterministic integrable input. 
All stochastic processes are defined on the filtered probability space ($\Omega, \cF, \{ \cF_t \}_{t\in[0,T]}, \mathbb{P}$), where $\omega \in \Omega$ denotes a sample result. For stochastic processes, we follow the standard subscript notation for time,  see e.g.,~\cite{oksendal2003stochastic, kloeden1992numerical, jazwinski2007stochastic}. For brevity, we omit the $\omega$ dependence and write $\widehat{\bX}_t$ in place of $\widehat{\bX}_t(\omega)$ throughout the paper.
The ROM~\eqref{eq:bilinear_stochastic_ROM} consists of the drift term that includes the linear operator $\widehat{\bA} \in \mathbb{R}^{r \times r}$, the input matrix $\widehat{\bB} \in \mathbb{R}^{r \times m}$, and the bilinear operators $\widehat{\bN}_{\ell} \in \mathbb{R}^{r \times r}$.
It also consists of the diffusion term with the coefficient matrix $\widehat{\bM} \in \mathbb{R}^{r \times d}$ that scales the stochastic forcing contribution to the ROM dynamics. Note that setting $\widehat{\bM} = \mathbf{0}$ eliminates the stochastic term and recovers the deterministic version of Eq.~\eqref{eq:bilinear_stochastic_ROM}, where $\widehat{\bX}_{t}$ reduces to a deterministic state~$\widehat{\bX}(t)$.
\par
The choice of model~\eqref{eq:bilinear_stochastic_ROM} is motivated by the following considerations.
First, in the experimental setting, the fluid surface dynamics are influenced by measurement noise, intrinsic experimental variability even under controlled operating conditions, and unresolved physical phenomena such as thermal effects on the fluid surface, independent of the acoustic forcing signal. These sources of uncertainty motivate modeling the reduced-order dynamics as an SDE.
Second, the experimental capillary wave system considered herein is driven by an acoustic sinusoidal signal that excites the piezoelectric transducer, inducing vibrations throughout the fluid layer to a depth below~$725 \, \upmu \textnormal{m}$.
While this forcing acts throughout the fluid, the DHM measurements are restricted to the two-dimensional region where the high-intensity light passes through.
Specifically, the measurement domain of~$300 \, \upmu \textnormal{m} \times 300 \, \upmu \textnormal{m}$ is much smaller than the full annular contact region between the fluid and the substrate, which has a diameter of $9{,}525 \, \upmu \textnormal{m}$ (see~\Cref{fig:DHM_setup_droplet_shcematic}).
Thus, within the observed measurement domain, the acoustic forcing cannot be regarded solely as a boundary excitation. Instead, it directly influences the measured fluid surface dynamics. Assuming this interaction carries over the reduced space, we include bilinear terms $\sum_{\ell=1}^{m} \widehat{\bN}_{\ell} \widehat{\bX}_{t} u_\ell(t)$, allowing the ROM to capture input-state coupling.
This yields a bilinear form of the right-hand side in the ROM. 
\par
While higher-order polynomial ROMs---including quadratic~\cite{McQuarrie_2021, Swischuk_2020, qian2022reduced, farcas2023parametric, KIM2025114418, Zastrow2023, 0x003d183f, Rocha2023} and cubic~\cite{jain2021performance} formulations---have been successfully developed and demonstrated for deterministic systems, analogous higher-order formulations for stochastic ROMs remain theoretically non-existent.
In particular, extending the stochastic OpInf framework~\cite{freitag2025learning} to higher-order polynomial SDE ROMs is nontrivial because nonlinear state terms introduce higher-order statistical moments into the expectation dynamics.
The expectation dynamics underlying the current framework are discussed later in~\Cref{proof:ROM_SDE_mean_dynamics_derivation}.
Specifically, in general,
$
\mathbb{E}[\widehat{\mathbf{X}}_t\otimes\widehat{\mathbf{X}}_t]
\neq \mathbb{E}[\widehat{\mathbf{X}}_t] \otimes \mathbb{E}[\widehat{\mathbf{X}}_t]
$ (here, the operator $\otimes$ denotes the \textit{compact} Kronecker product). Consequently, the statistical evolution becomes significantly more complicated.
We therefore adopt a bilinear formulation as the ROM that captures the capillary wave dynamics most expressively.

\begin{assumption}
Each experimental trajectory is treated as an independent realization of the underlying stochastic dynamics. 
\end{assumption}
\noindent The stochastic OpInf framework requires multiple realizations of the same stochastic process to learn the ROM.
In the demonstration cases of~\cite{freitag2025learning}, these realizations are generated by repeatedly simulating a known stochastic high-fidelity model with different numerically-generated noise realizations.
In the experimental setting considered herein, however, the same approach is not possible.
Therefore, we treat repeated experimental measurements acquired under identical operating conditions as independent realizations of the same underlying stochastic dynamics.

\begin{assumption}
The fluid surface displacement data are approximately stationary over some time interval~$t_d$. 
\end{assumption}
\noindent This allows us to partition a single long experimental trajectory of duration $T_{f} = 0.868\,\textnormal{s}$ into multiple non-overlapping segments of length $t_d$, each of which is treated as a sample trajectory of the underlying stochastic process (see~\cite{kossaifi2026demystifying} for a similar approach). Although non-overlapping, adjacent segments may retain temporal dependence~\cite{Priestley1981,Oliver2014}.
To determine an approximate stationary time-window length~$t_d$, we partition the full time frame~$T_f$ into $N_s = \left \lfloor \frac{T_f}{t_d} \right \rfloor$ non-overlapping segments of duration~$t_d$, where $\lfloor \cdot \rfloor$ denotes the floor function. Then, the number of measurement samples is determined as
\begin{equation}
    L = N_s \times N_{\textnormal{rep}},
\label{eq:sample_size}
\end{equation}
where $N_{\textnormal{rep}}$ is the number of repeated experiments for the given input case.
To select $t_d$, we assess the physical and statistical consistency across the partitioned sample trajectories \textcolor{magenta}{through} their power spectral densities (PSDs), which provide a natural characterization of the spectral energy distribution in capillary wave turbulence~\cite{newell2011wave}.
\begin{figure}[!ht]
\centering
\includegraphics[width=1.0\linewidth]{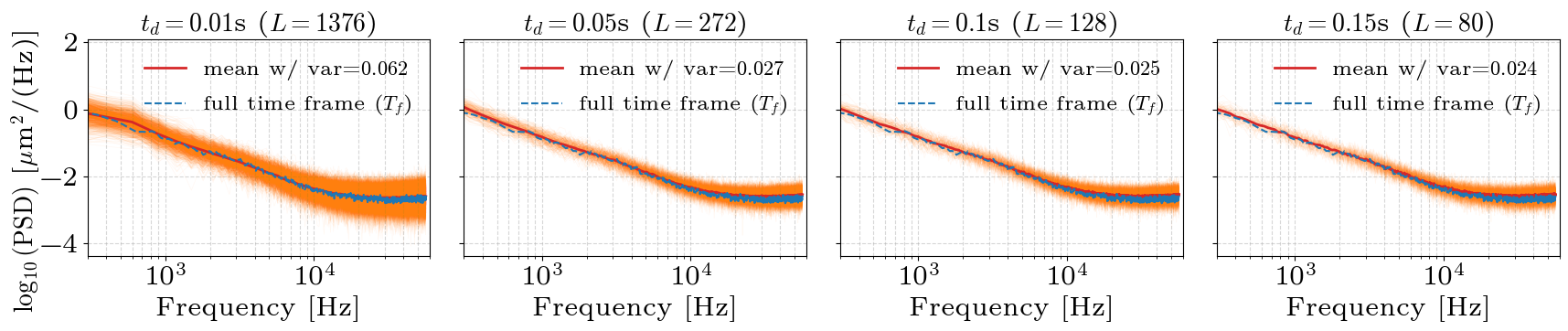}
\caption{Comparison of power spectral densities of sample trajectories when different partitioning time periods~$t_d$ are considered. The case of~$\textnormal{Ca}_\textnormal{ac} = 0.00284$ is shown here. Power spectral densities are estimated via Welch's method~\cite{Welch1967}, using the time series at the midpoint of the 2D spatial domain.}
\label{fig:PSD_decorrelation_time}
\end{figure}
\Cref{fig:PSD_decorrelation_time} shows the comparison of PSDs of all sample trajectories when we partition the full time frame using different time intervals of length $t_d$. Each figure shows~$L$ PSDs for the tested~$t_d$ and their mean. For reference, the PSD computed over the full time frame~$T_f$ is also shown.
For all partition lengths considered, the mean PSD computed across the $L$ PSD samples matches closely with the PSD computed from the full unpartitioned trajectory.
In contrast, the PSD variance across samples depends more strongly on the choice of $t_d$. Specifically, the variance changes marginally from $0.024$ to $0.025$ and $0.027$ as $t_d$ decreases from $0.15 \, \textnormal{s}$ to $0.10 \, \textnormal{s}$ and $0.05 \, \textnormal{s}$. However, it increases substantially to $0.062$ when the shortest partition length of $t_d=0.01 \, \textnormal{s}$ is considered, which shows the expected tradeoff between segment length and spectral variability in Welch-type methods~\cite{Welch1967}.
We select $t_d=0.05 \, \textnormal{s}$ as it achieves a favorable balance between spectral consistency and the number of independent trajectories we obtain, without introducing a substantial spectral variability.
With this strategy, \Cref{tab:experimental_datasets_capillary_wave} summarizes the experimental datasets used for each $\textnormal{Ca}_\textnormal{ac}$ case, from which the stochastic ROM is learned using $L$ measurement samples.
Among all the repeated datasets described in~\Cref{ss:DHM_experimental_data}, we only use datasets that exhibit consistent spectral behaviors---i.e., verified by checking PSDs---for the stochastic ROM modeling.
\begin{table}[t]
    \centering
    \caption{Experimental datasets of deionized water surface displacements used for the stochastic ROM.}
    \vspace{-0.5em}
    
    \begin{tabular}{l|rrrrrrrrrr}
        $\textnormal{Acoustic capillary number}$ $(\times 10^{-3})$ & 0.25 & 0.66 & 0.67 & 0.83 & 1.15 & 1.56 & 1.76 & 1.99 & 2.84 & 3.32 \\
        \hline
        $\textnormal{Number of repeated experiments}$ & 16 & 15 & 14 & 15 & 16 & 15 & 16 & 15 & 16 & 15 \\
        $\textnormal{Number of measurement samples}$ & 272 & 255 & 238 & 255 & 272 & 255 & 272 & 255 & 272 & 255 \\
    \end{tabular}
\label{tab:experimental_datasets_capillary_wave}
\end{table}

%%%%%%%%%%%%%%%%%%%%%% ROM SDE and its mean/covariance dynamics %%%%%%%%%%%%%%%%%%%%%%%%%
\subsection{Mean and covariance dynamics of the reduced states}
\label{ss:ROM_SDE_mean_covariance_dynamics_derivation}
We treat the DHM measurements of the spatiotemporal fluid surface displacement as high-fidelity state trajectories.  
Then, we formulate the stochastic OpInf ROM framework directly in the reduced space; we omit the derivation in the high-dimensional space presented in the original work.
Since the bilinear stochastic ROM~\eqref{eq:bilinear_stochastic_ROM} is linear in the reduced state for a prescribed input and driven by the additive Wiener process~$\bW_t$, its solution is a Gaussian process
$$
\widehat{\bX}_t \sim \cN(\widehat{\bE}(t), \widehat{\bC}(t)),
$$
characterized by the mean~$\widehat{\bE}(t) = \mathbb{E}\left[ \widehat{\bX}_{t} \right] \in \real^r$ and covariance~$\widehat{\bC}(t) = \Cov\left( \widehat{\bX}_{t}, \widehat{\bX}_{t} \right) = \mathbb{E}\left[ \widehat{\bX}^c_{t} (\widehat{\bX}^{\,c}_{t})^\top \right] \in \real^{r \times r}$, where $\widehat{\bX}^c_{t} \coloneqq \widehat{\bX}_{t} - \widehat{\bE}(t) \in \real^{r}$ is the centered process (\cite[Chapter.~4.2]{kloeden1992numerical}.
The proof for the following proposition is based on that of~\cite[Sect.~2]{freitag2025learning}; we include it here for completeness.

%%%%% Proposition %%%%%%
\begin{proposition}[Mean dynamics]
    The mean of the reduced states of Eq.~\eqref{eq:bilinear_stochastic_ROM}, $\widehat{\bE}(t) \in \real^r$, admits the following evolution equation:
    \begin{align}
        \dot{\widehat{\bE}}(t) =
        \widehat{\bA} \widehat{\bE}(t) + \widehat{\bB} \bu(t) + \sum_{\ell=1}^{m} \widehat{\bN}_{\ell} \widehat{\bE}(t) u_\ell(t).
    \label{eq:ROM_expectation_dynamics}
    \end{align}
\end{proposition}
\par
\begin{proof}
\label{proof:ROM_SDE_mean_dynamics_derivation}
Formulating the SDE ROM~\eqref{eq:bilinear_stochastic_ROM} in integral form yields
\begin{equation}
    \widehat{\bX}_{t} - \widehat{\bX}_{0} = \int_{0}^{t} \left[ \widehat{\bA} \widehat{\bX}_{s} + \widehat{\bB} \bu(s) + \sum_{\ell=1}^{m} \widehat{\bN}_{\ell} \widehat{\bX}_{s} u_\ell(s) \right] \mathrm{d}s + \int_0^t \widehat{\bM} \mathrm{d} \bW_s,
\label{eq:integral_form_of_ROM_SDE}
\end{equation}
where $\int_0^t \widehat{\bM} \mathrm{d} \bW(s)$ denotes the It\^o integral~\cite{ito1944109}.
Taking expectations on both sides of Eq.~\eqref{eq:integral_form_of_ROM_SDE} and applying the zero-mean property of the It\^o integral, the mean trajectory $\widehat{\bE}(t) \in \real^r$ of the reduced states satisfies
\begin{equation*}
    \widehat{\bE}(t) = \widehat{\bE}(0) + \int_0^t \left[ \widehat{\bA} \widehat{\bE}(s) + \widehat{\bB} \bu(s) + \sum_{\ell=1}^{m} \widehat{\bN}_{\ell} \widehat{\bE}(s) u_\ell(s) \right] \mathrm{d}s,
\end{equation*}
which in differential form reads as
\begin{equation*}
    \dot{\widehat{\bE}}(t) = \widehat{\bA} \widehat{\bE}(t) + \widehat{\bB} \bu(t) + \sum_{\ell=1}^{m} \widehat{\bN}_{\ell} \widehat{\bE}(t) u_\ell(t).
% \label{eq:ROM_expectation_dynamics}
\end{equation*}
\end{proof}
\par
\noindent The covariance dynamics are also derived in~\cite[Sect.~2]{freitag2025learning}. Here, we present a more detailed derivation, making explicit several intermediate steps that are omitted in the original work.
%
%
%%%%% Proposition %%%%%%
\begin{proposition}[Covariance dynamics]
    The covariance matrix of the reduced states of Eq.~\eqref{eq:bilinear_stochastic_ROM}, $\widehat{\bC}(t) \in \real^{r \times r}$, admits the following evolution equation:
    \begin{align}
        \dot{\widehat{\bC}}(t) =
        \widehat{\bP}(t) \widehat{\bC}(t) + \widehat{\bC}(t) {\widehat{\bP}(t)}^\top + \widehat{\bM} \bK \widehat{\bM}^\top,
    \label{eq:ROM_covariance_dynamics}        
    \end{align}
    where $\widehat{\bP}(t) \coloneqq \widehat{\bA} + \sum_{\ell=1}^{m} \widehat{\bN}_{\ell} u_\ell(t) $ and $\widehat{\bC}(0) = \Cov(\widehat{\bX}_{0}, \widehat{\bX}_{0}) \coloneqq \mathbb{E}[\widehat{\bX}_{0} \widehat{\bX}_{0}^\top] - \mathbb{E}[\widehat{\bX}_{0}]\mathbb{E}[\widehat{\bX}_{0}]^\top$.
\end{proposition}
\par
\begin{proof}
\label{proof:ROM_SDE_covariance_dynamics_derivation}
We start by writing the SDE for the centered process~$\mathrm{d} \widehat{\bX}^c_{t}= \mathrm{d} \left( \widehat{\bX}_{t} - \widehat{\bE}(t) \right)$ using Eq.~\eqref{eq:ROM_expectation_dynamics} as:
\begin{equation}
\begin{split}
    \mathrm{d} \widehat{\bX}^c_{t}
    &= \left[ \widehat{\bA} \widehat{\bX}_{t} + \widehat{\bB} \bu(t) + \sum_{\ell=1}^{m} \widehat{\bN}_{\ell} \widehat{\bX}_{t} u_\ell(t) \right] \mathrm{d}t + \widehat{\bM} \mathrm{d} \bW_t - \left[ \widehat{\bA} \bE_r(t) + \widehat{\bB} \bu(t) + \sum_{\ell=1}^{m} \widehat{\bN}_{\ell} \bE_r(t) u_\ell(t) \right] \mathrm{d}t \\
    &= \left[ \widehat{\bA} \widehat{\bX}^c_{t} + \sum_{\ell=1}^{m} \widehat{\bN}_{\ell} \widehat{\bX}^c_{t} \bu(t)  \right] \mathrm{d}t + \widehat{\bM} \mathrm{d} \bW_t.
\end{split}
\label{eq:centered_process_ROM_SDE}
\end{equation}
Applying It\^o's product rule, we expand
\begin{equation}
    \rd \widehat{\bX}^c_{t} \, (\widehat{\bX}^{c}_{t})^\top = \underbrace{\rd \widehat{\bX}^c_{t} \, (\widehat{\bX}^{c}_{t})^\top}_{(\text{i})} + \underbrace{\widehat{\bX}^c_{t} \, \rd(\widehat{\bX}^{c}_{t})^\top}_{(\text{ii})} + \underbrace{\rd \widehat{\bX}^c_{t} \,  \rd(\widehat{\bX}^{c}_{t})^\top}_{(\text{iii})},
\label{eq:Ito_product_rule_ROM}
\end{equation}
 where the last term on the right-hand side is the It\^o correction arising from the quadratic variation $\rd \bW_t \, \rd \bW_t^\top = \bK \rd t$.
Using Eq.~\eqref{eq:centered_process_ROM_SDE}, each term in Eq.~\eqref{eq:Ito_product_rule_ROM} can be expanded as follows:
 \begin{align*}
\text{(i)}: \quad \rd \widehat{\bX}^c_{t} \, (\widehat{\bX}^{\,c}_{t})^\top 
    &= \left[\left[ \widehat{\bA} \widehat{\bX}^c_{t} + \sum_{\ell=1}^{m} \widehat{\bN}_{\ell} \widehat{\bX}^c_{t} u_\ell(t)\right] \, \rd t + \widehat{\bM} \rd \bW_t\right] (\widehat{\bX}^{c}_{t})^\top \\
    &= \left[ \widehat{\bA} \widehat{\bX}^c_{t} (\widehat{\bX}^{c}_{t})^\top + \sum_{\ell=1}^{m} \widehat{\bN}_{\ell} \widehat{\bX}^c_{t} (\widehat{\bX}^{c}_{t})^\top u_\ell(t) \right] \rd t + \widehat{\bM} \, \rd \bW_t \, (\widehat{\bX}^{c}_{t})^\top, \\
\text{(ii)}: \quad \widehat{\bX}^c_{t} \, \rd (\widehat{\bX}^{\,c}_{t})^\top 
    &= \widehat{\bX}^c_{t} \left[\left[ \widehat{\bA} \widehat{\bX}^c_{t} + \sum_{\ell=1}^{m} \widehat{\bN}_{\ell} \widehat{\bX}^c_{t} u_\ell(t) \right ] \rd t + \widehat{\bM} \rd \bW_t \right]^\top \\
    &= \left[ \widehat{\bX}^c_{t} (\widehat{\bX}^{c}_{t})^\top \widehat{\bA}^\top + \sum_{\ell=1}^{m} \widehat{\bX}^c_{t} (\widehat{\bX}^{c}_{t})^\top \widehat{\bN}_{\ell}^\top u_\ell(t) \right ] \rd t + \widehat{\bX}^c_{t} \, \rd \bW_t^\top \, \widehat{\bM}^\top, \\
\text{(iii)}: \quad \rd \widehat{\bX}^c_{t} \, \rd (\widehat{\bX}^{\,c}_{t})^\top 
    &= \left[\left[ \widehat{\bA} \widehat{\bX}^c_{t} + \sum_{\ell=1}^{m} \widehat{\bN}_{\ell} \widehat{\bX}^c_{t} u_\ell(t)\right ] \rd t + \widehat{\bM} \, \rd \bW_t \right] \\
    &\quad \times
       \left[\left[ \widehat{\bA} \widehat{\bX}^c_{t} + \sum_{\ell=1}^{m} \widehat{\bN}_{\ell} \widehat{\bX}^c_{t} u_\ell(t)\right] \rd t + \widehat{\bM} \, \rd \bW_t \right]^\top \\
    &= \widehat{\bM} \, \rd \bW_t \, \rd \bW_t^\top \, \widehat{\bM}^\top \\
    &= \widehat{\bM} \bK \widehat{\bM}^\top \rd t.
\end{align*}
Note, It\^o's rules $(\rd t)^2=0$ and $\rd \bW_t \cdot \rd t = 0$ are used during the expansion of term (iii).
As a result, Eq.~\eqref{eq:Ito_product_rule_ROM} can be rewritten as:
\begin{equation}
\begin{split}
    \rd \widehat{\bX}^c_{t} (\widehat{\bX}^{c}_{t})^\top
    &= \left[ \widehat{\bA} \widehat{\bX}^c_{t} (\widehat{\bX}^{c}_{t})^\top + \sum_{\ell=1}^{m} \widehat{\bN}_{\ell} \widehat{\bX}^c_{t} (\widehat{\bX}^{c}_{t})^\top u_\ell(t) + \widehat{\bX}^c_{t} (\widehat{\bX}^{c}_{t})^\top \widehat{\bA}^\top + \sum_{\ell=1}^{m} \widehat{\bX}^c_{t} (\widehat{\bX}^{c}_{t})^\top \widehat{\bN}_{\ell}^\top u_\ell(t) \right] \rd t      \\
    &\quad + \widehat{\bM} \, \rd \bW_t \, (\widehat{\bX}^{c}_{t})^\top + \widehat{\bX}^c_{t} \, \rd \bW_t^\top \, \widehat{\bM}^\top + \widehat{\bM} \bK \widehat{\bM}^\top \rd t.
\end{split}
\label{eq:centered_process_ROM_SDE_expanded}
\end{equation}
Taking the expectation to~Eq.~\eqref{eq:centered_process_ROM_SDE_expanded}, which applies to terms that have stochastic processes, we obtain the SDE for~$\widehat{\bC}(t) = \mathbb{E}[\widehat{\bX}^c_{t} (\widehat{\bX}^{c}_{t})^\top] \in \real^{r \times r}$ as
\begin{equation*}
\begin{split}
    \rd \widehat{\bC}(t) 
    &\coloneqq \rd \left( \mathbb{E} \left[ \widehat{\bX}^c_{t} (\widehat{\bX}^{c}_{t})^\top \right] \right)  \\
    &= \mathbb{E} \left[ \widehat{\bA} \widehat{\bX}^c_{t} (\widehat{\bX}^{c}_{t})^\top \rd t + \widehat{\bX}^c_{t} (\widehat{\bX}^{c}_{t})^\top \widehat{\bA}^\top \rd t + \sum_{\ell=1}^{m} \left( \widehat{\bN}_{\ell} \widehat{\bX}^c_{t} (\widehat{\bX}^{c}_{t})^\top u_\ell(t) + \widehat{\bX}^c_{t} (\widehat{\bX}^{c}_{t})^\top \widehat{\bN}_{\ell}^\top u_\ell(t) \right) \rd t \right]  \\ 
    &\quad + \mathbb{E} \left[ \widehat{\bM} \, \rd \bW_t \, (\widehat{\bX}^{c}_{t})^\top \right] + \mathbb{E} \left[ \widehat{\bX}^c_{t} \, \rd \bW_t^\top \, \widehat{\bM}^\top \right] + \widehat{\bM} \bK \widehat{\bM}^\top \rd t   \\
    &= \left[ \widehat{\bA} \widehat{\bC}(t) + \widehat{\bC}(t) \widehat{\bA}^\top + \sum_{\ell=1}^{m} \left( \widehat{\bN}_{\ell} u_\ell(t) \widehat{\bC}(t) + \widehat{\bC}(t) \widehat{\bN}_{\ell}^\top u_\ell(t) \right) \right] \rd t + \widehat{\bM} \bK \widehat{\bM}^\top \rd t.
\end{split}
% \label{eq:ROM_covariance_derivation}
\end{equation*}
This reads in differential form as a matrix differential equation:
\begin{equation*}
    \dot{\widehat{\bC}}(t) = \widehat{\bP}(t) \widehat{\bC}(t) + \widehat{\bC}(t) {\widehat{\bP}(t)}^\top + \widehat{\bM} \bK \widehat{\bM}^\top,
% \label{eq:ROM_covariance_dynamics}
\end{equation*}
where $\widehat{\bP}(t) \coloneqq \widehat{\bA} + \sum_{\ell=1}^{m} \widehat{\bN}_{\ell} u_\ell(t) $ and $\widehat{\bC}(0) = \Cov(\widehat{\bX}_{0}, \widehat{\bX}_{0}) \coloneqq \mathbb{E}[\widehat{\bX}_{0} \widehat{\bX}_{0}^\top] - \mathbb{E}[\widehat{\bX}_{0}]\mathbb{E}[\widehat{\bX}_{0}]^\top.$
\end{proof}
\par
\noindent Notably, although the reduced states $\widehat{\bX}_t$ evolve according to a SDE, the mean dynamics~\eqref{eq:ROM_expectation_dynamics} are deterministic; the driving Wiener process $\bW_t$ does not appear.
The covariance dynamics~\eqref{eq:ROM_covariance_dynamics} are likewise deterministic but depend explicitly on the diffusion coefficient matrix $\widehat{\bM}$ through the It\^o correction term~$\widehat{\bM} \bK \widehat{\bM}^\top$.
% From a physical perspective, this term encodes the intensity of the noise injected into the stochastic system.

%%%%%%%%%%%%%%%%%%%%%% Stochastic OpInf %%%%%%%%%%%%%%%%%%%%%%%%%%%%
\subsection{Stochastic differential equation reduced-order model learning problems}
\label{ss:SDE_ROMs_learning_for_capillary_wave_system}

%%%%%%%%%%%%%%%%%%%%%% Drift OpInf problem %%%%%%%%%%%%%%%%%%%%%%%%%%%%
\subsubsection{Learning the drift operators}
\label{sss:drift_opinf_problem}
The operators in the drift term of the ROM SDE~\eqref{eq:bilinear_stochastic_ROM} fully govern the mean dynamics of the reduced states in Eq.~\eqref{eq:ROM_expectation_dynamics}.
In practice, the mean of the reduced states is generally not directly observable. Thus, it must be estimated from the finite number of samples of reduced states.
Given $L$ noise samples at each discrete time instance~$t_k$, $k=0,\ldots,s$, the sample mean of the reduced states, which approximates the expected value of the reduced state trajectories $\mathbb{E}[\bPhi_r^\top \bX(t_k)]$, is computed as
\begin{equation}
    \widehat{\bE}_{k}^{L} = \frac{1}{L} \sum_{j=1}^{L} \widehat{\bX}(t_k, \omega_j) \in \real^r,
\label{eq:empirical_mean_reduced_state}
\end{equation}
which is used as an empirical estimate of the mean state trajectory governed by Eq.~\eqref{eq:ROM_expectation_dynamics}.
Our first goal is to learn the unknown drift operators in Eq.~\eqref{eq:bilinear_stochastic_ROM} using the empirical mean of the reduced states. Following~\cite{freitag2025learning}, we refer to this learning problem as the \textit{drift OpInf problem}, which is solved via the least-squares problem
\begin{equation}
    \min_{\widehat{\bA}, \widehat{\bB}, \widehat{\bN}_{1}, \ldots, \widehat{\bN}_{m}} \sum_{k=0}^{s} \left\| \widehat{\bA} \widehat{\bE}_{k}^{L} + \widehat{\bB} \bu_k + \sum_{\ell=1}^{m} \widehat{\bN}_{\ell} \widehat{\bE}_{k}^{L} u_{\ell, k} - \dot{\widehat{\bE}}^{\raisebox{-0.8ex}{\scriptsize$L$}}_k  \right\|_2^2,
\label{eq:drift_OpInf_optimization}
\end{equation}
where $\bu_k \coloneqq \bu(t_k) \in \real^m$ and $u_{\ell,k} \coloneqq u_\ell(t_k) \in \real$.
Equation~\eqref{eq:drift_OpInf_optimization} can be written compactly as
\begin{equation}
    \min_{\bO} \left\| \bR^\top - \bD^\top \bO^\top \right\|_F^2,
\label{eq:drift_OpInf_compact_form}
\end{equation}
where the target operators~$\bO = [ \widehat{\bA} \quad \widehat{\bB} \quad \widehat{\bN}_{1}, \ldots, \widehat{\bN}_{m} ] \in \real^{r \times (r + m + mr)}$.
$\bR = [ \dot{\widehat{\bE}}^{\raisebox{-0.8ex}{\scriptsize$L$}}_0, \ldots, \dot{\widehat{\bE}}^{\raisebox{-0.8ex}{\scriptsize$L$}}_s ] \in \real^{r \times (s+1)}$, where $\dot{\widehat{\bE}}^{\raisebox{-0.8ex}{\scriptsize$L$}}_k$ denotes the time derivative of~$\widehat{\bE}_{k}^L$, estimated using a suitable finite difference method. In this work, the 6th-order central finite differencing is used. 
The data matrix~$\bD$ is constructed as~$\left[ (\widehat{\bE}^L)^\top \quad {\bU}^\top \quad (\bU \odot \widehat{\bE}^L)^\top \right]^\top \in \real^{(r + m + mr) \times (s+1)}$, where $\widehat{\bE}^L = \left[ \widehat{\bE}_{0}^{L}, \ldots, \widehat{\bE}_{s}^{L} \right] \in \real^{r \times (s+1)}$, $\bU = \left[ \bu_0, \ldots, \bu_s \right] \in \real^{m \times (s+1)}$, and $\bU \odot \widehat{\bE}^L = \left[ \bu_0 \otimes \widehat{\bE}_{0}^{L}, \ldots, \bu_s \otimes \widehat{\bE}_{s}^{L} \right] \in \real^{mr \times (s+1)}$.
Here, a \textit{compact} Kronecker product $\otimes$ calculates the element-wise multiplication of two vectors while removing redundant quadratic terms. For example, when $\ba = [a_1, a_2, a_3]$, the compact Kronecker product is computed as $\ba \otimes \ba = [a_1^2, a_1a_2, a_1a_3, a_2^2, a_2a_3, a_3^2]$. 
Solving the drift OpInf problem~\eqref{eq:drift_OpInf_compact_form} identifies the operators~$\widehat{\bA}$, $\widehat{\bB}$, and $\widehat{\bN}_{1}, \ldots, \widehat{\bN}_{m}$, which describe the deterministic component of our ROM dynamics.

\subsubsection{Learning the diffusion operator}    \label{sss:diffusion_opinf_problem}
The next learning problem identifies the diffusion coefficient matrix~$\widehat{\bM}$, which modulates the intensity of the stochastic forcing term in the ROM~\eqref{eq:bilinear_stochastic_ROM}.
However, the covariance of the reduced states is also not directly observable. Thus, it must be estimated from the reduced state samples, analogous to the fact that the sample mean is estimated from the finite number of available samples in Eq.~\eqref{eq:empirical_mean_reduced_state}.
Using all the samples of the reduced states in~\eqref{eq:reduced_state_snapshot}, the covariance of the reduced states is approximated~by
\begin{equation}
    \widehat{\bC}_{k}^{L} = \frac{1}{L-1} \sum_{j=1}^{L} \left[ \widehat{\bX}(t_k, \omega_j) - \widehat{\bE}_{k}^{L} \right] \left[ \widehat{\bX}(t_k, \omega_j) - \widehat{\bE}_{k}^{L} \right]^\top \in \real^{r \times r},
\label{eq:empirical_covariance_reduced_state}
\end{equation}
which provides an empirical estimate of the covariance trajectory governed by Eq.~\eqref{eq:ROM_covariance_dynamics}.
The second learning problem, called the \textit{diffusion OpInf problem}, then seeks a constant matrix~$\widehat{\bH} \in \real^{r \times r}$ that best fits the residual covariance dynamics unexplained by the drift terms, i.e., the diffusion contribution via:
\begin{equation}
    \min_{\widehat{\bH} \in \real^{r \times r}} \sum_{k=0}^{s} \left\| \dot{\widehat{\bC}}_{k}^L - \left[ \widehat{\bP}_{k} \widehat{\bC}_{k}^L + \widehat{\bC}_{k}^L {\widehat{\bP}_{k}}^\top \right] - \widehat{\bH} \right\|_F^2,
\label{eq:diffusion_OpInf_optimization}
\end{equation}
where $\widehat{\bP}_{k} \coloneqq \widehat{\bP}(t_k) = \widehat{\bA} + \sum_{\ell=1}^{m} \widehat{\bN}_{\ell} u_{\ell}(t_k)$.
Notably, the system matrix of the covariance dynamics, $\widehat{\bP}(t)$, is defined as the sum of linear and bilinear operators, which are learned via the drift OpInf problem~\eqref{eq:drift_OpInf_compact_form}.
The time derivative of~$\widehat{\bC}_{k}^L$ is estimated using a 2nd-order central finite difference method in this work.
Ultimately, we obtain the diffusion operator by matrix decomposition of~$\widehat{\bH} = \widehat{\bM} \bK \widehat{\bM}^\top$. 
Note that the solution matrix~$\widehat{\bH}$ obtained via the least-squares problem~\eqref{eq:diffusion_OpInf_optimization} is not necessarily symmetric positive semidefinite. Thus, we seek to obtain the symmetric positive semidefinite approximant of $\widehat{\bH}$, using the method introduced in~\cite{higham1988computing}.
This is first done by symmetrizing $\widehat{\bH}$ by $\widehat{\bH}^{\text{sym}} = \frac{\widehat{\bH} + \widehat{\bH}^\top}{2}$, and computing eigendecomposition~$\widehat{\bH}^{\text{sym}} = \bQ \mathbf{\Lambda} \bQ^\top$, where the columns of~$\bQ$ are eigenvectors of the symmetrized $\widehat{\bH}^{\text{sym}}$ and $\mathbf{\Lambda}$ is a diagonal matrix containing the eigenvalues of $\widehat{\bH}^{\text{sym}}$. Then we clamp negative eigenvalues to zero to obtain the symmetric positive semidefinite approximant of~$\widehat{\bH}$ as $\widehat{\bH}^{\succeq 0} = \bQ \mathbf{\Lambda_+} {\bQ}^\top$, where $\mathbf{\Lambda}_+ = \text{diag}(\text{max}(\lambda, 0))$.
\subsubsection{Summary of the stochastic OpInf ROM learning framework}
\label{sss:summary_stochastic_OpInf_ROM}
The complete stochastic ROM learning framework is illustrated in~\Cref{fig:stochastic_OpInf_framework}. 
In the experimental phase, high-fidelity DHM measurements are acquired and verified before ROM construction. The consistency of the measurement samples is assessed by checking their spectral contents, such as the PSDs, to ensure that the data adequately represent the underlying stochastic dynamics. If inconsistencies among the datasets are detected, additional experiments are performed. 
Among the verified measurements, the one-dimensional profiles extracted at the center of the spatial domain are used to assemble the state snapshot matrix~\eqref{eq:high_fidelity_state_snapshot}. We compute the SVD of~\eqref{eq:high_fidelity_state_snapshot} to obtain the POD basis, onto which the high-dimensional states are projected to obtain the reduced state snapshots~\eqref{eq:reduced_state_snapshot}. These reduced states are then used to estimate their mean and covariance, from which the dynamical equations for the mean~\eqref{eq:ROM_expectation_dynamics} and covariance~\eqref{eq:ROM_covariance_dynamics} are constructed. Two optimization problems~\eqref{eq:drift_OpInf_optimization} and~\eqref{eq:diffusion_OpInf_optimization} are subsequently solved in sequence to learn the drift and diffusion operators of the stochastic ROM. Finally, the learned ROM is simulated using an appropriate numerical integrator (e.g., the implicit Euler--Maruyama scheme), and the resulting reduced state trajectories are lifted to the high-dimensional state space.
\begin{figure}[htbp]
\centering
\includegraphics[width=0.85\linewidth]{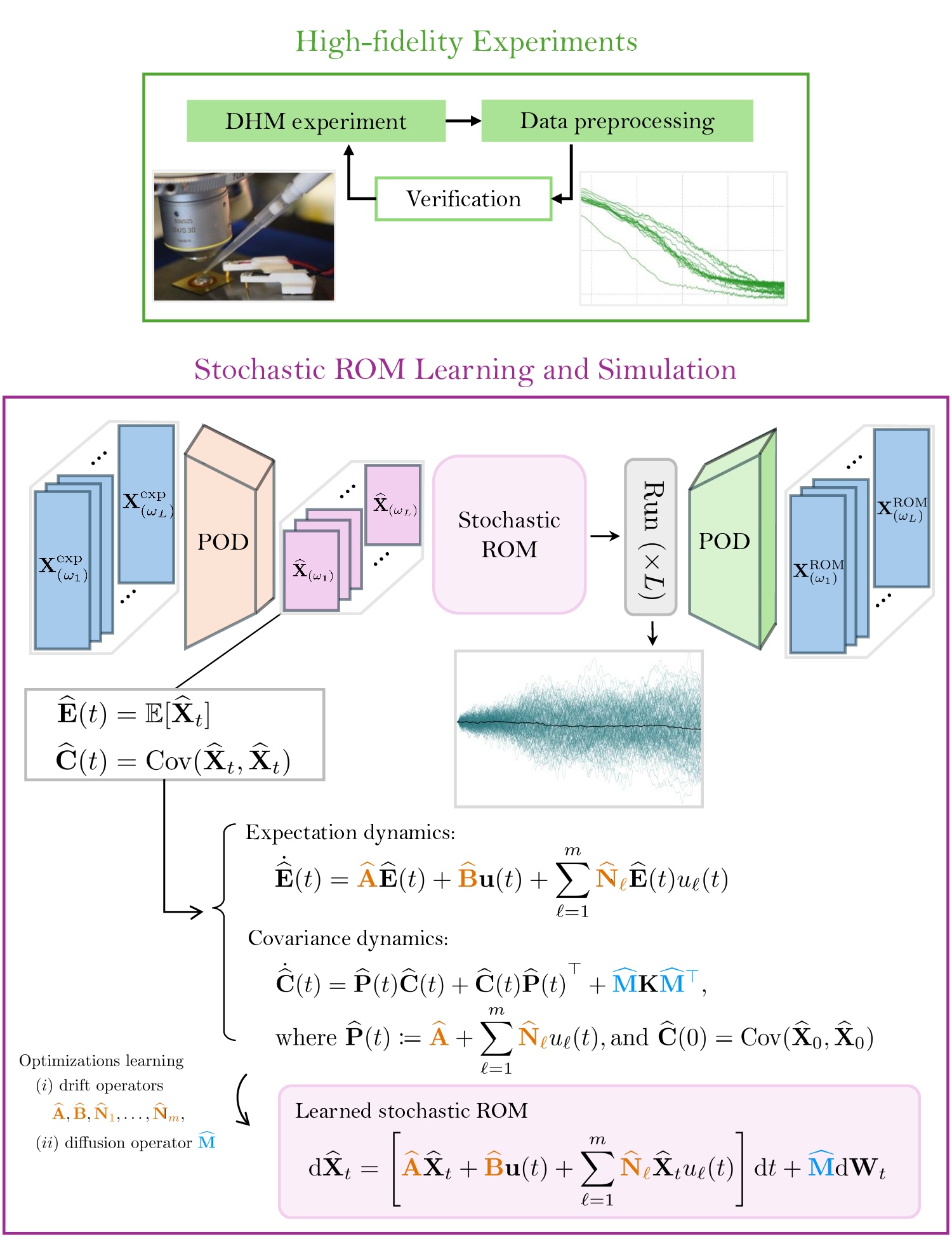}
\caption{Schematic of the stochastic OpInf ROM learning framework. High-fidelity DHM measurements are first verified to ensure that they consistently represent the underlying stochastic dynamics. Each block $\bX_{(\omega_j)}^{\textnormal{exp}}$, where $j=1,\ldots,L$, represents the trajectory of experimental measurement snapshots associated with the $j$th noise realization, i.e., $\bX_{(\omega_j)}^\textnormal{exp} = [\bX(t_0,\omega_j), \ldots, \bX(t_s,\omega_j)]$, see Eq.~\eqref{eq:high_fidelity_state_snapshot}.
The reduced state snapshots $\widehat{\bX}_{\omega_j}$ are obtained by projecting the high-dimensional states onto a POD basis. The reduced states are then used to estimate the mean and covariance, from which two sequential optimization problems~\eqref{eq:drift_OpInf_compact_form} and~\eqref{eq:diffusion_OpInf_optimization} learn the drift and diffusion operators of the stochastic ROM. Finally, the learned ROM is simulated and the reduced state trajectories are lifted to the high-dimensional state space.}
\label{fig:stochastic_OpInf_framework}
\end{figure}

\subsubsection{Evaluation of the learned ROM via mean and covariance errors}    \label{sss:stochastic_ROM_error_metrics}
Once the drift and diffusion operators are learned, we use the implicit Euler--Maruyama scheme to simulate the learned ROM SDE~\eqref{eq:bilinear_stochastic_ROM}.
We simulate the ROM using the initial conditions of all $L$ samples, obtained via projection of all high-dimensional samples in Eq.~\eqref{eq:high_fidelity_state_snapshot}. These initial conditions are denoted as~$\mathbb{X}_{r,0}$ in the state snapshot matrix~\eqref{eq:reduced_state_snapshot}.
For experimental data, the number of realizations of initial conditions is limited by the available measurements, and additional samples require collecting new experimental realizations or extrapolating from fitted distributions. 
In our case, the stochastic ROM operators are learned from the available $L$ reduced state samples obtained from the DHM measurement samples, determined as in Eq.~\eqref{eq:sample_size}.
Furthermore, we note that the noise realizations $\bW_t$, which are drawn with a fixed random number generator seed during the ROM simulations, are not directly accessible from the DHM measurements. Therefore, obtaining a ROM that achieves a pathwise (strong) approximation of the original state is not feasible.
Instead, we aim to find the ROM that approximates the original states in a weak sense: the learned ROM is considered accurate if its state distribution matches that from the original reduced states at each time step $t \in [0, T]$.
In practice, we assess this by computing the first two moments---mean and covariance---and the corresponding errors are defined as:
\begin{equation}
    \varepsilon_{\textnormal{mean}} \coloneqq \frac{\sum_{k=0}^{s}  \left \| \widehat{\bE}_{k}^L - \widehat{\bE}^{L,\textnormal{ROM}}_{k} \right \|_2^2}{\sum_{k=0}^{s} \left \| \widehat{\bE}_{k}^{L} \right \|_2^2}, \quad
    \varepsilon_{\textnormal{cov}} \coloneqq \frac{\sum_{k=0}^{s}  \left \| \widehat{\bC}_{k}^L -  \widehat{\bC}^{L,\textnormal{ROM}}_{k} \right \|_F^2}{\sum_{k=0}^{s} \left \| \widehat{\bC}_{k}^{L} \right\|_F^2}.
\label{eq:weak_mean_covaraince_error}
\end{equation}    
Here, $\widehat{\bE}_{k}^{L, \textnormal{ROM}}$ and $\widehat{\bC}_{k}^{L, \textnormal{ROM}}$ are the empirical mean and covariance, estimated by the $L$ ROM states at the $k$th time step. They are computed by Eqs.~\eqref{eq:empirical_mean_reduced_state} and~\eqref{eq:empirical_covariance_reduced_state} but using the ROM results.
% \begin{align}
% \widehat{\bE}_{k}^{L, \textnormal{ROM}} &= \frac{1}{L}\sum_{j=1}^{L} \widehat{\bX}_{k}^{\textnormal{ROM}}(\omega_j), 
% \label{eq:empirical_mean}
% \\    
% \widehat{\bC}_{k}^{L, \textnormal{ROM}} &= \frac{1}{L-1} \sum_{j=1}^{L} \left[ \widehat{\bX}_{k}^{\textnormal{ROM}}(\omega_j) - \widehat{\bE}_{k}^{L,\textnormal{ROM}} \right ] \left[ \widehat{\bX}_{k}^{\textnormal{ROM}}(\omega_j) - \widehat{\bE}_{k}^{L,\textnormal{ROM}} \right ]^\top,
% \label{eq:empirical_covariance}
% \end{align}
% where $\widehat{\bX}_{k}^{\textnormal{ROM}}(\omega_j) \in \real^{r}$ is the ROM state for the $j$th noise sample, at the time step $k$.

%%%%%%%%%%%%%%%%%%%%%% Regularization hyperparameters selection %%%%%%%%%%%%%%%%%%%%%%%%%%%
\subsection{Regularization hyperparameters associated with spectral contents of ROM trajectories}    \label{ss:regularization_hyperparameters_selection_based_on_spectral_contents}
Within the deterministic OpInf ROM framework, Tikhonov regularization has been widely used to improve numerical accuracy and stability across many applications~\cite{McQuarrie_2021, mcquarrie2023nonintrusive, issan2023predicting, SAWANT2023115836, qian2022reduced, KIM2025114418, geelen2023learning, kang2026}.
In contrast, the original work on stochastic OpInf~\cite{freitag2025learning} did not apply regularization in its benchmark problems, leaving its role in the stochastic setting unexplored.
We incorporate Tikhonov regularization into the stochastic OpInf framework and show that, beyond improving numerical accuracy and stability, it carries additional physical significance: the regularization hyperparameters affect the spectral content of the resulting ROM state trajectories.
This role differs between the two OpInf learning problems: regularizing the deterministic drift terms helps mitigate aliasing effects, while regularizing the stochastic diffusion term modulates the stochastic energy injected into the ROM, which we discuss further next.

%%%%%%%%%%%%%%%%%%%%%% Regularization of bilinear operator %%%%%%%%%%%%%%%%%%%%%%%%%%%
\subsubsection{Mitigation of aliasing effects through bilinear operators regularization}
\label{sss:mitigate_aliasing}
There are two driving forces in the DHM system: one is the deterministic input signal that excites the transducer, and the other is the stochastic process noise present in the system.
Specifically, our input signal is a scalar (i.e., the input dimension $m=1$), given by $u(t)=A_f \cos(2\pi f_0 t)$ with frequency $f_0 = 7 \times 10^6\,\text{Hz}$.
For the notation, in what follows we write the input as the scalar $u(t)$ (and $u_k$ at discrete time $t_k$), rather than the general vector-valued $\bu(t)$ used earlier, and accordingly drop the subscript $\ell$ on $\widehat{\bN}$ since only a single bilinear operator $\widehat{\bN}$ is needed.
Consequently, our ROM becomes
\begin{equation}
    \mathrm{d} \widehat{\bX}_{t} = \left[\widehat{\bA} \widehat{\bX}_{t} + \widehat{\bB} u(t) + \widehat{\bN} \widehat{\bX}_{t} u(t) \right] \mathrm{d} t + \widehat{\bM} \mathrm{d} \bW_t.
\label{eq:ROM_SDE_scalar_input}
\end{equation}
The DHM camera has a sampling rate of $f_s = 115{,}200\,\text{Hz}$, yielding the Nyquist frequency of $f_N = f_s/2 = 57{,}600\,\text{Hz}$.
When we simulate the ROM using the corresponding time step of $1/115{,}200 \approx 8.68 \times 10^{-6} \, \textnormal{s}$, the input signal---well above the Nyquist frequency $f_N$---is under-resolved, resulting in aliasing.
Since the ROM SDE~\eqref{eq:bilinear_stochastic_ROM} is formulated in continuous time, it could in principle be simulated with a much finer time step that resolves our high-frequency input signal.
However, this requires a proportionally larger number of time steps within the fixed time frame. In our case, this increases the online simulation cost in proportion to the downsampling ratio, undermining the computational efficiency that motivates the ROM in the first place.
The aliasing frequency is given by
\begin{equation}
    f_{\text{alias}} = |f_0 - \kappa f_s|,
\label{eq:aliasing_frequency}
\end{equation}
where the integer $\kappa$ is chosen such that $f_{\text{alias}} \in [0, f_N]$.
In our case, $f_0/f_s=7{,}000{,}000\,\textnormal{Hz}/115{,}200\,\textnormal{Hz} \approx 60.76$, so the nearest integer is $\kappa=61$. This yields the aliasing frequency of $f_{\text{alias}}=|f_0 - 61f_s| = 27{,}200 \, \textnormal{Hz}$.
\par
To mitigate this effect, we apply Tikhonov regularization to the drift OpInf problem~\eqref{eq:drift_OpInf_optimization}:
\begin{equation}
    \min_{\widehat{\bA}, \widehat{\bB}, \widehat{\bN}} 
    \sum_{k=0}^{s} \left\| \widehat{\bA} \widehat{\bE}_{k}^{L} + \widehat{\bB} u_k 
    + \widehat{\bN} \widehat{\bE}_{k}^{L} u_k 
    - \dot{\widehat{\bE}}_{k}^{L} \right\|_2^2
    + \gamma_1 \|\widehat{\bA}\|_F^2 
    + \gamma_2 \|\widehat{\bB}\|_F^2 
    + \gamma_3 \|\widehat{\bN}\|_F^2,
\label{eq:drift_OpInf_with_regularization}
\end{equation}
where $\gamma_1$, $\gamma_2$, $\gamma_3$ penalize the drift operators $\widehat{\bA}$, $\widehat{\bB}$, $\widehat{\bN}$, respectively.
Note that although the excitation amplitude $A_f$ prescribed by the signal generator (see~\Cref{fig:DHM_setup_droplet_shcematic}) varies across experimental cases, we fix the amplitude of the ROM input signal $u(t)$ to $A_f = 1$ when solving Eq.~\eqref{eq:drift_OpInf_with_regularization} for all cases in~\Cref{tab:experimental_datasets_capillary_wave}.
This is because the effective forcing amplitude acting on the measured fluid surface domain is unknown. Specifically, the input signal is amplified before driving the piezoelectric transducer, and although the particle velocity at the transducer is directly measured by the laser Doppler vibrometer, the transfer of this actuation to the DHM measurement domain through the fluid layer cannot be characterized.
Therefore, we let the regularized drift OpInf problem~\eqref{eq:drift_OpInf_with_regularization} learn the input operator~$\widehat{\bB}$, which absorbs the unknown scaling between the prescribed input signal and its effective influence on the ROM dynamics.
Since the aliased frequency falls within the resolvable band, it can consequently enter the simulated dynamics through the input-dependent terms $\widehat{\bB} u(t)$ and $\widehat{\bN} \widehat{\bX}_{t} u(t)$ in the drift, causing spurious content to appear in the frequency spectrum of the ROM results. 
Regularizing the operators $\widehat{\bB}$ and $\widehat{\bN}$ suppresses the injection of aliased energy into the ROM, thereby mitigating its effect on the ROM results' spectral content.
In our case, penalizing the bilinear operator $\widehat{\bN}$ showed more substantial impact on aliasing mitigation.
\par
\Cref{fig:aliasing_N_regs_comparison} compares the ROM's spectral content across different levels of $\widehat{\bN}$ regularization.
Regularizing the bilinear operator with $\gamma_3= 10^1, 10^2, 10^3, 10^4, 10^5$ produces a monotonic decrease in its Frobenius norm, from $7.97 \times 10^4$ to $3.33 \times 10^4$, $1.83 \times 10^3$, $6.73 \times 10^1$, and $8.20 \times 10^{-1}$.
Over the same range of $\gamma_3$, the weak mean error decreases sharply from $1.53 \times 10^{-1}$ to $5.57 \times 10^{-2}$ at $\gamma_3=10^2$, then plateaus for $\gamma_3 = 10^3, 10^4, 10^5$ at $5.55 \times 10^{-2}$, $6.12 \times 10^{-2}$, and $6.13 \times 10^{-2}$, respectively.
The PSDs, computed from the ROM results obtained under different levels of regularization for $\widehat{\bN}$, show a spurious peak at $f_\textnormal{alias} = 27{,}200\,\textnormal{Hz}$ (computed via Eq.~\eqref{eq:aliasing_frequency}) for $\gamma_3 = 10^1$ and $10^2$, while this peak does not appear for $\gamma_3 = 10^3, 10^4, 10^5$.
Notably, $\gamma_3 = 10^2$ achieves a weak mean error comparable to the larger regularization cases, yet its PSD still exhibits the aliasing artifact.
This shows that, even though the ROM indicates good accuracy in terms of the weak mean error, spurious aliased frequencies can still be present in the resulting ROM trajectories. 
In practice, $\gamma_3$ is selected by minimizing the weak mean error, searched over a wide initial range, e.g., $\gamma_3 \in [10^0, 10^{10}]$.
If the selected ROM generates state trajectories that exhibit aliasing artifacts in the frequency domain, we narrow the search range---e.g., $\gamma_3 \in [10^5, 10^{10}]$, or even more restrictively to $[10^8, 10^{10}]$---and reselect $\gamma_3$ within this restricted range. We begin with a wide range to first explore the hyperparameter space broadly, narrowing it only when aliasing artifacts are observed.
\begin{figure}[!ht]
    \centering
    \begin{subfigure}[t]{0.49\linewidth}
        \centering
        \includegraphics[width=\linewidth]{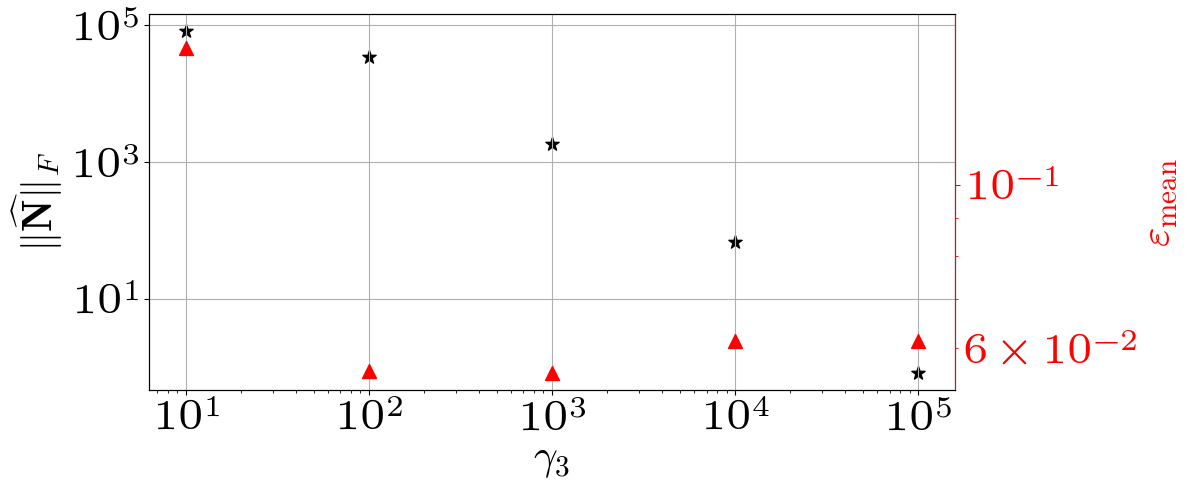}
        % \caption{}
        \label{fig:aliasing_N_r_norm_weak_mean_error_comparison}
    \end{subfigure}
    \hfill
    \begin{subfigure}[t]{0.49\linewidth}
        \centering
        \includegraphics[width=\linewidth]{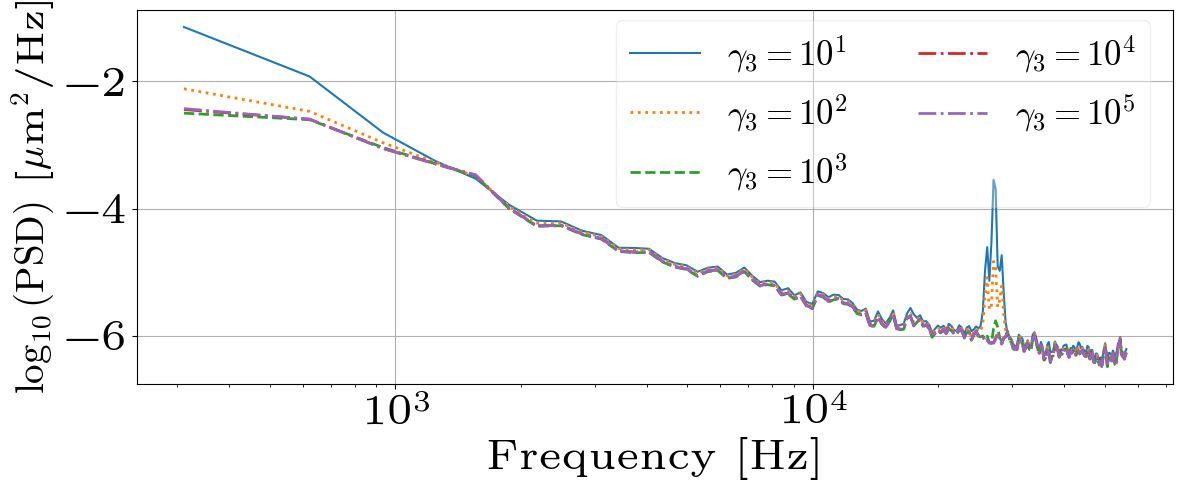}
        % \caption{}
        \label{fig:aliasing_PSD_comparison}
    \end{subfigure}
    \vspace{-1em}
    \caption{Left: Comparison of the Frobenius norms of the learned bilinear operators and the weak mean errors, defined in Eq.~\eqref{eq:weak_mean_covaraince_error}, across different levels of penalization. Right: Comparison of the corresponding PSDs under different penalization levels. The PSDs are estimated via Welch’s method, using the mean time series at the midpoint of the 1D spatial domain.}
    \label{fig:aliasing_N_regs_comparison}
\end{figure}

\begin{figure}[!ht]
    \centering
    \includegraphics[width=0.75\linewidth]{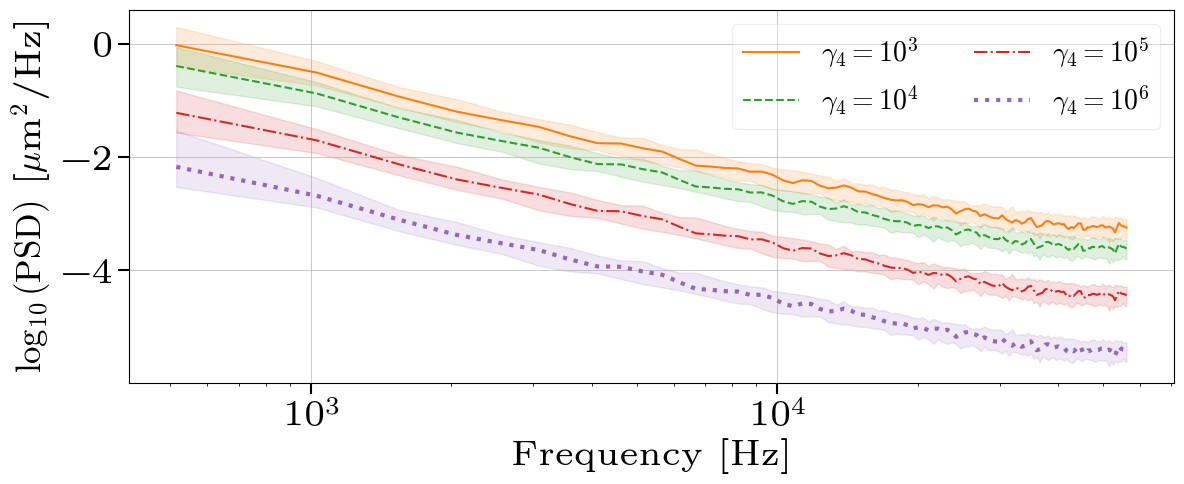}
    \caption{Comparison of the spectral content of ROM trajectories, generated under different levels of regularization for the diffusion term in Eq.~\eqref{eq:diffusion_OpInf_with_regularization}.
    The PSDs are estimated via Welch's method, using $L=272$ reconstructed ROM trajectory samples at the midpoint of the 1D spatial domain. For each case, the mean of the $272$ PSDs is shown along with the corresponding band.}
    \label{fig:Diffusion_regularization_comparison}
\end{figure}

%%%%%%%%%%%%%%%%%%%%%% Regularization of diffusion operator %%%%%%%%%%%%%%%%%%%%%%%%%%%
\subsubsection{Modulation of ROM stochastic energy through diffusion operator regularization}
\label{sss:modulate_PSDs}
While the discussion above focused on regularizing the input-related terms---especially the bilinear term---to mitigate aliasing, regularization of the diffusion contribution serves a distinct role: it modulates the stochastic energy injected into the ROM.
To this end, we apply Tikhonov regularization to the diffusion OpInf problem~\eqref{eq:diffusion_OpInf_optimization}:
\begin{equation}
    \min_{\widehat{\bH}} \sum_{k=0}^{s} \left\| \dot{\widehat{\bC}}_{k}^L - \left[ \widehat{\bP}_{k} \widehat{\bC}_{k}^L + \widehat{\bC}_{k}^L {\widehat{\bP}_{k}}^\top \right] - \widehat{\bH} \right\|_F^2 + \gamma_4 \left\| \widehat{\bH} \right \|_F^2,
\label{eq:diffusion_OpInf_with_regularization}
\end{equation}
where $\gamma_4$ penalizes the operator $\widehat{\bH}$ that includes the diffusion operator $\widehat{\bM}$ before factorization, see~\Cref{sss:diffusion_opinf_problem}.
The operators learned by solving the regularized drift and diffusion OpInf problems of Eqs.~\eqref{eq:drift_OpInf_with_regularization} and~\eqref{eq:diffusion_OpInf_with_regularization} are used to learn our SDE ROM operators in~\eqref{eq:ROM_SDE_scalar_input}.
In particular, the diffusion-term regularization hyperparameter $\gamma_4$ affects the norm of the learned diffusion operator $\widehat{\bM}$.
Since $\widehat{\bM}$ acts as the stochastic forcing in our ROM~\eqref{eq:ROM_SDE_scalar_input}, this in turn directly affects the stochastic energy injected into the system. 
To isolate this effect, we vary $\gamma_4$ while holding all other modeling knobs fixed: the ROM dimension $r=14$, the linear-term regularization $\gamma_1 = 10^3$, the input-term regularization $\gamma_2 = 10^6$, and the bilinear-term regularization $\gamma_3 = 10^8$.
Increasing $\gamma_4$ across $10^3, 10^4, 10^5, 10^6$ causes the Frobenius norm of the learned diffusion operator $\widehat{\bM}$ to decrease monotonically from $1.23 \times 10^5$ to $8.06 \times 10^4$, $3.11 \times 10^4$, and $1.01 \times 10^4$.
\Cref{fig:Diffusion_regularization_comparison} compares the resulting spectral energy, computed from $L$ generated ROM trajectories, across these same four levels of diffusion-term regularization.
As $\gamma_4$ increases and the norm of $\widehat{\bM}$ decreases, the overall spectral content of the ROM trajectories decreases correspondingly.
This shows that $\gamma_4$ directly modulates the stochastic energy injected into the learned ROM.
Together with the weak mean and covariance errors discussed in~\Cref{sss:stochastic_ROM_error_metrics}, these physically interpretable roles motivate considering spectral contents for the regularization strategy: while $\gamma_3$ is selected via grid search as part of the ROM dimension selection procedure that we discuss next, $\gamma_4$ is chosen heuristically, guided by the extent to which the ROM captures spectral energy of the experimental data.
%

%%%%%%%%%%%%%%%%%%%%%% ROM dimension selection %%%%%%%%%%%%%%%%%%%%%%%%%%%
\subsection{ROM dimension selection via constrained multi-objective optimization}    \label{ss:r_selection_based_on_mean_cov_error_tradeoff}
In the POD approach, the accuracy of a ROM critically depends on the quality of the basis, i.e., how well it captures the high-fidelity numerical simulation or experimental data. 
Selecting a proper ROM dimension $r$ in Eq.~\eqref{eq:reduced_state_snapshot} is thus a central step in constructing an accurate ROM~\eqref{eq:bilinear_stochastic_ROM}.
The most common approach to determine this is the energy criterion, informed by the singular values of the snapshot matrix in Eq.~\eqref{eq:high_fidelity_state_snapshot}, see, e.g.,~\cite{buchan2015pod}.
Specifically, when the basis matrix~$\mathbf{\Phi}$ is used, the projection error is defined as
\begin{equation*}
    \cE_{\text{proj}}(r) = \frac{\left\| \mathbb{X} - \bPhi_r \bPhi_r^\top \mathbb{X} \right\|_F^2}{\left\| \mathbb{X} \right\|_F^2} = 1 - \cE_{\text{cum}}(r),
% \label{eq:projection_cumulative_error}
\end{equation*}
where $\bPhi_r \bPhi_r^\top \in \mathbb{R}^{n \times n}$ is a linear projection onto the $r$-dimensional subspace spanned by the basis vectors (columns of~$\bPhi_r$). 
Here, the cumulative energy $\cE_{\text{cum}}(r)$ is defined as the ratio of the sum of squared singular values up to rank $r$ to the total sum of squared singular values:
\begin{equation*}
    \cE_{\text{cum}}(r) = \frac{\sum_{\eta=1}^{r} \sigma_\eta^2}{\sum_{\eta=1}^{n} \sigma_\eta^2},
% \label{eq:cumulative_energy}
\end{equation*}
where the $\sigma_\eta$ are the singular values of the state snapshot matrix~$\mathbb{X}$.
However, even for many complex systems---including the data considered herein---the first one or two modes can capture over 99\% of the total energy, see also~\cite{deane1991low}.
Consequently, the conventional energy-based criterion becomes almost uninformative and provides little practical guidance for selecting the ROM dimension; any choice of $r \geq 2$ satisfies commonly used energy thresholds.
\begin{table}[htbp]
    \centering
    \caption{Cumulative energy $\cE_{\text{cum}}(r)$ captured by the first $r$ POD basis for different acoustic forcing cases $\mathrm{Ca}_{\mathrm{ac}}$.}
    \vspace{-0.5em}
    \begin{tabular}{r|llllll}
        $\mathrm{Ca}_{\mathrm{ac}}$ & $r=1$ & $r=2$ & $r=3$ & $r=4$ & $r=5$ & $r=6$ \\
        \hline
        $0.25 \times 10^{-3}$ & 0.999696 & 0.999931 & 0.999995 & 0.999996 & 0.999997 & 0.999997 \\
        $0.66 \times 10^{-3}$ & 0.998053 & 0.999656 & 0.999986 & 0.999990 & 0.999998 & 0.999998 \\
        $0.67 \times 10^{-3}$ & 0.997179 & 0.999705 & 0.999983 & 0.999990 & 0.999998 & 0.999999 \\
        $0.83 \times 10^{-3}$ & 0.998446 & 0.999798 & 0.999990 & 0.999998 & 0.999999 & 0.999999 \\
        % $1.15 \times 10^{-3}$ & 0.997846 & 0.999680 & 0.999979 & 0.999993 & 0.999995 & 0.999996 \\
        % $1.56 \times 10^{-3}$ & 0.998352 & 0.999767 & 0.999978 & 0.999993 & 0.999995 & 0.999996 \\
        % $1.76 \times 10^{-3}$ & 0.999640 & 0.999944 & 0.999990 & 0.999995 & 0.999996 & 0.999997 \\
        % $1.99 \times 10^{-3}$ & 0.999829 & 0.999960 & 0.999984 & 0.999989 & 0.999991 & 0.999992 \\
        % $2.84 \times 10^{-3}$ & 0.999916 & 0.999974 & 0.999985 & 0.999989 & 0.999991 & 0.999992 \\
        % $3.32 \times 10^{-3}$ & 0.999907 & 0.999971 & 0.999983 & 0.999987 & 0.999989 & 0.999991
    \end{tabular}
    \label{tab:all_cases_cumulative_energy_comparison}
\end{table}
\Cref{tab:all_cases_cumulative_energy_comparison} confirms this behavior for the four representative experimental cases shown; the remaining cases exhibit the same trend.
\Cref{fig:singular_values_cum_energy_0p08} further shows the rapid decay of the singular values, along with the fast saturation of the cumulative energy and the corresponding low projection errors achieved by only a few leading modes, for the case $\textnormal{Ca}_\textnormal{ac} = 0.67 \times 10^{-3}$.
\begin{figure}[h]
    \centering

    \begin{subfigure}[t]{0.49\linewidth}
        \vspace{0pt}
        \centering
        \includegraphics[width=\linewidth]{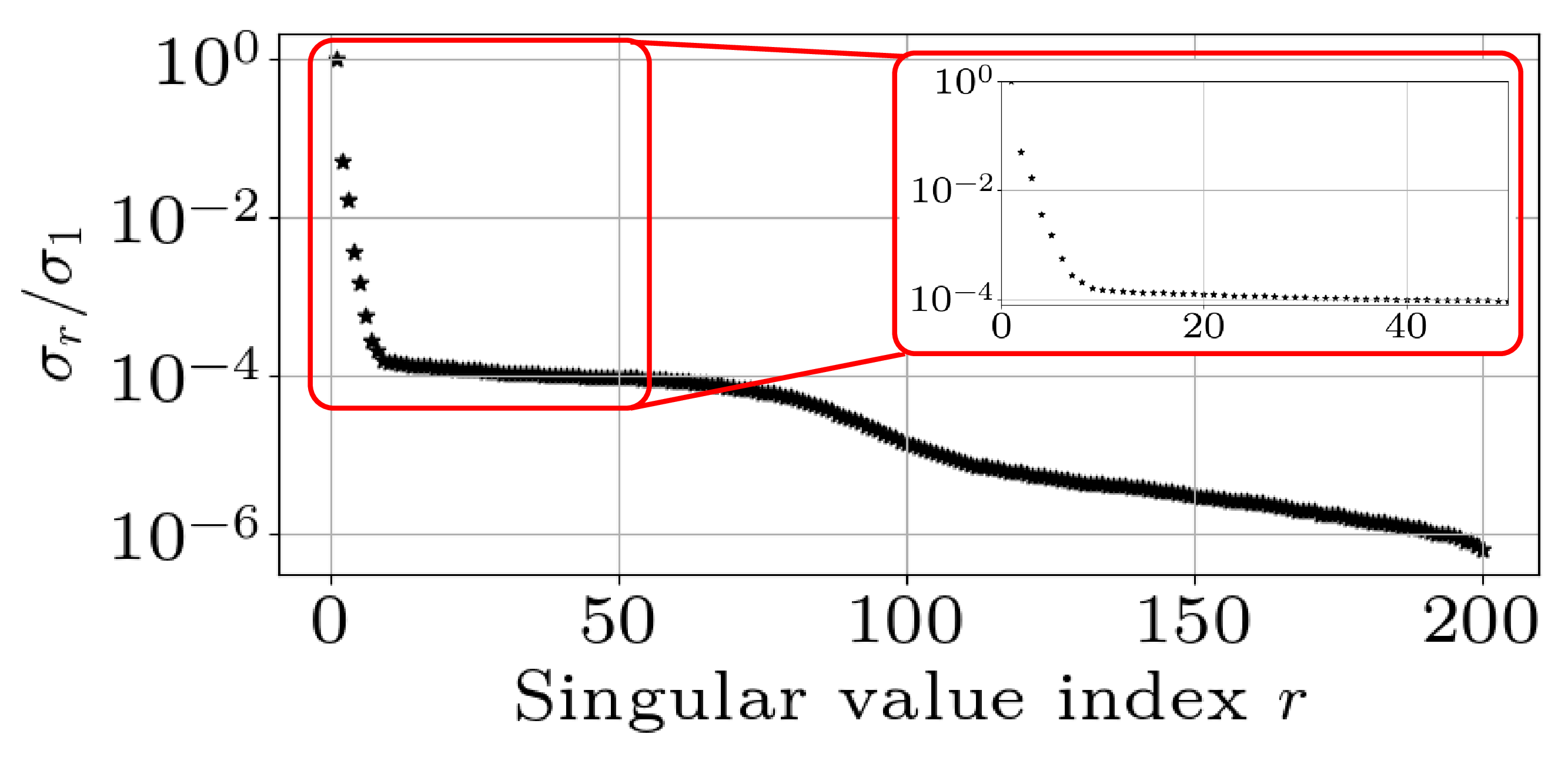}
        % \caption{}
        \label{fig:normalized_singular_values_0p08}
    \end{subfigure}
    \hfill
    \begin{subfigure}[t]{0.49\linewidth}
        \vspace{0pt}
        \centering
        \includegraphics[width=\linewidth]{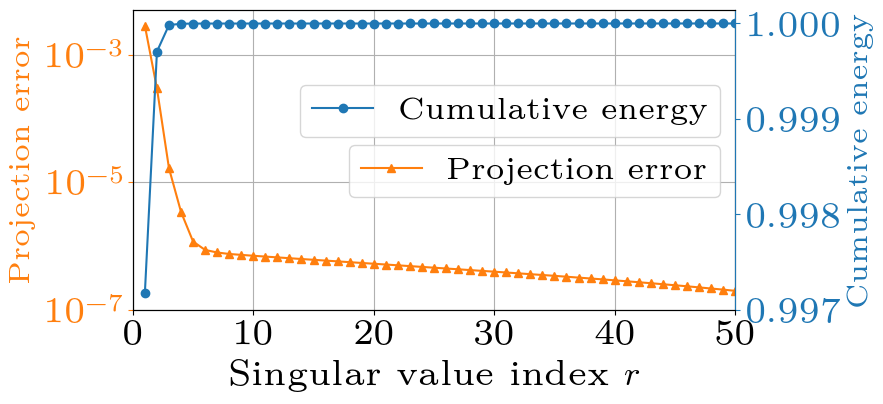}
        % \caption{}
        \label{fig:projection_cumulative_energy_0p08}
    \end{subfigure}
    \vspace{-1em}
    \caption{The case of $\textnormal{Ca}_\textnormal{ac}=0.67 \times 10^{-3}$. Left: Decay of the normalized singular values up to the full rank $n$. Right: Projection error and cumulative energy of the POD basis for ranks up to $r=50$.}
    \label{fig:singular_values_cum_energy_0p08}
\end{figure}
More importantly, satisfying a prescribed energy threshold does not guarantee downstream ROM accuracy.
As also observed in~\cite{deane1991low}, ROM accuracy need not improve monotonically with increasing dimension; a basis that captures nearly all of the snapshot energy may still produce an inaccurate ROM.
These limitations have been recognized in the literature.
The authors in~\cite{bhattacharyya2022experimental} propose an energy-closure criterion to complement situations where the simple energy-based approach for dimension selection is insufficient.
A related limitation arises when the POD basis is constructed in a parametric setting: the energy criterion provides no guarantee that the basis generalizes to different parameter regimes. The work~\cite{hay2009local} addresses this via sensitivity analysis, augmenting the basis to improve robustness across the parameter domain.
In the context of OpInf specifically, many prior works have selected the ROM dimension based on the energy criterion~\cite{Swischuk_2020, farcas2024domain, qian2019transform, farcas2023parametric}.
The authors in~\cite{geelen2022localized}, however, mention that the energy criterion is only used for snapshot clustering and classifier training in their localized ROM framework---not as a proxy for actual ROM accuracy.
The study~\cite{qian2022reduced} demonstrates that increasing the ROM dimension can increase the risk of overfitting, highlighting a trade-off: larger ROM dimensions offer greater expressivity, but lower-dimensional ROMs are often better able to learn accurate dynamics from limited training data.
\par
We argue that these considerations carry over in the stochastic setting. While the original stochastic OpInf study~\cite{freitag2025learning} observed that both the mean and covariance errors generally decreased and converged with increasing $r$ for the benchmark one-dimensional bilinear and two-dimensional linear stochastic heat equations, such well-behaved convergence cannot be expected to hold for general nonlinear systems, particularly for more complex systems such as the one considered herein.
This consideration is particularly important in our setting, where the ROM operators are identified through two consecutive learning problems: the drift-term learning and the diffusion-term learning problems.
Motivated by this, we propose a criterion for determining the ROM dimension $r$ that accounts for both mean and covariance accuracy.
The proposed approach consists of two steps. First, for each candidate ROM dimension $r$, we select the regularization hyperparameters $\bgamma=(\gamma_1,\gamma_2,\gamma_3,\gamma_4)$ that minimize the weak mean error.
Note that $\widehat{\bE}^{L,\textnormal{ROM}}_{k}$ and $\widehat{\bC}^{L,\textnormal{ROM}}_{k}$ in Eq.~\eqref{eq:weak_mean_covaraince_error} depend implicitly on $r$ and $\bgamma$; in what follows we make this dependence explicit by writing the corresponding errors $\varepsilon_{\textnormal{mean}}(r;\bgamma)$ and $\varepsilon_{\textnormal{cov}}(r;\bgamma)$, since the selection procedure requires comparing these errors across ROM dimension candidates and regularization parameters.
Minimizing the weak error over the regularization parameters for each ROM dimension is formulated as
\begin{equation*}
    \bgamma^{*}(r) = \operatorname*{arg\,min}_{\bgamma} \, \varepsilon_{\textnormal{mean}}(r; \bgamma).
% \label{eq:hyperparameters_minimizing_mean_error}
\end{equation*}
Using the resulting mean-optimal hyperparameters~$\bgamma^{*}(r)$, we compare weak covariance errors for all candidate ROM dimensions.
The ROM dimension is then selected by minimizing the weak covariance error among candidates whose mean error remains within a prescribed margin of the best mean error across all candidates.
This is formally expressed as
\begin{equation}
    r^{*} = \operatorname*{arg\,min}_{r} \varepsilon_{\textnormal{cov}}(r; \bgamma^{*}(r)), \quad \textnormal{subject to} \quad \varepsilon_{\textnormal{mean}}(r; \bgamma^{*}(r)) \leq \min_r(\varepsilon_\textnormal{mean}(r ; \bgamma^{*}(r))) + \tau_\textnormal{tol},
\label{eq:r_selection_criterion}
\end{equation}
where $\tau_\textnormal{tol}$ is a user-specified tolerance that controls how far the mean error of a candidate $r$ may deviate from the best achievable mean error.
This formulation ensures that the selected ROM dimension achieves good covariance accuracy without sacrificing mean accuracy, while accounting for regularization effects at each candidate dimension.
\begin{figure}[h]
\centering
\begin{subfigure}[t]{0.49\linewidth}
    \centering
    \includegraphics[width=\linewidth]{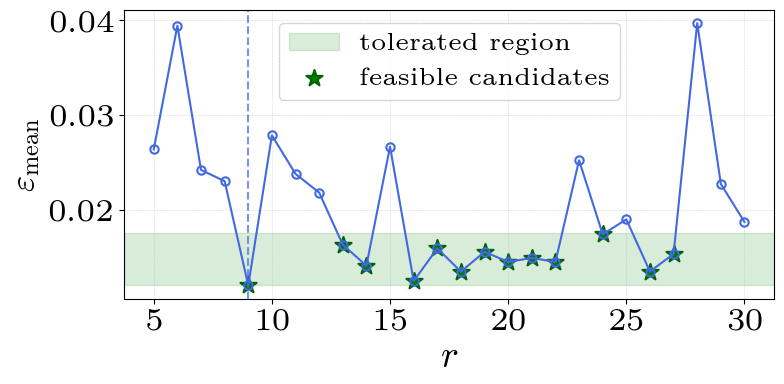}
    % \caption{}
    \label{fig:r_selection_mean_error}
\end{subfigure}
% \hfill
\begin{subfigure}[t]{0.49\linewidth}
    \centering
    \includegraphics[width=\linewidth]{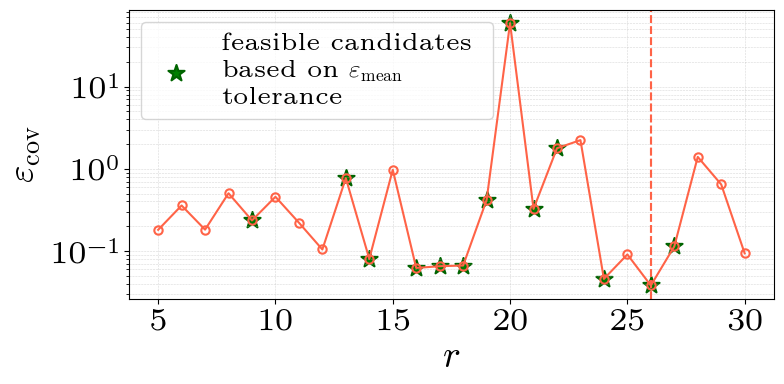}
    % \caption{}
    \label{fig:r_selection_covariance_error}
\end{subfigure}
\vspace{-1em}
\caption{Weak mean and covariance errors, computed using Eq.~\eqref{eq:weak_mean_covaraince_error}, are shown over ROM dimension candidates. The case $\textnormal{Ca}_\textnormal{ac} = 1.76 \times 10^{-3}$ is considered.
Left: Feasible ROM dimensions that satisfy the selection criterion in Eq.~\eqref{eq:r_selection_criterion} are shown. The shaded region indicates the tolerated region, and the dashed vertical line marks the $r$ that minimizes the weak mean error alone. Right: Weak covariance error over ROM dimensions. The dashed vertical line indicates the final selected ROM dimension among the feasible candidates.}
\label{fig:r_selection_based_on_mean_cov_errors}
\end{figure}
\par
\Cref{fig:r_selection_based_on_mean_cov_errors} illustrates the proposed selection process by showing the weak mean and covariance errors across all tested ROM dimension candidates $r = 5, 6, \ldots, 30$.
We can observe that the two errors do not necessarily exhibit the same trend over $r$: increasing the ROM dimension may decrease the mean error while simultaneously increasing the covariance error, and the dimensions that individually minimize each error do not generally coincide.
This highlights an inherent trade-off in ROM dimension selection between accurately capturing the first- and second-order moments of the reduced state.
Moreover, the ROM dimensions that individually minimize these two errors do not generally coincide, further highlighting the trade-off between accurately capturing first- and second-order moments of the reduced states.
This motivates our Pareto-inspired approach, which selects the ROM dimension by first enforcing a tolerance on the weak mean error, then minimizing the weak covariance error over the resulting feasible candidates.
In our case, the tolerance $\tau_\textnormal{tol}$ is defined as 
\begin{equation*}
    \tau_\textnormal{tol} = \beta \, \left( \max_r(\varepsilon_{\textnormal{mean}}(r; \bgamma^{*})) - \min_r(\varepsilon_{\textnormal{mean}}(r; \bgamma^{*})) \right),
% \label{eq:tolerance}
\end{equation*}
where $0 \leq \beta \leq 1 $ specifies the fraction of the observed variation in the weak mean error used to define the tolerance. 
For example, $\beta = 0$ corresponds to zero tolerance, admitting only the ROM dimension that yields the minimum mean error, while $\beta = 1$ admits all candidates. More realistically, $\beta = 0.2$ corresponds to a tolerance of 20\% of the range of the weak mean error across all $r$ candidates.
\par
Note that, although in this particular case the selected value $r=26$ coincides with the global minimum of the covariance error over all tested dimensions, this behavior is not guaranteed in general.
It is also possible that none of the candidate dimensions satisfying the mean error tolerance achieves a particularly low covariance error. 
Even in such cases, the proposed criterion still selects the candidate with the smallest covariance error among those that satisfy the prescribed mean accuracy requirement. This ensures that the reconstructed mean trajectory remains accurate, which is essential because it directly determines the downstream performance metrics, such as mean wave fields, PSD, bispectrum, and bicoherence, which we discuss further in~\Cref{sec:numerical_results}.
We also note that this approach incurs a higher computational offline ROM training cost than the standard strategy of energy-based dimension selection followed by separate hyperparameter tuning. Nonetheless, it offers a more systematic framework for selecting the \textit{optimal} ROM dimension given our criteria, while avoiding manual iteration over $r$-candidates and directly incorporating model performance and regularization effects into the selection process.
%
%

%%%%%%%%%%%%%%%%%%%%%%%%%%%% Numerical results %%%%%%%%%%%%%%%%%%%%%%%%%%%%%%%%
\section{Numerical results} \label{sec:numerical_results}
We present the numerical results of our study on stochastic OpInf ROMs for the experimental capillary wave turbulence data using several metrics.
The first three analyses assess how well the learned ROMs reconstruct the statistics and spectral characteristics of the data used for learning. In contrast, the final analysis evaluates their predictive capability on held-out, unseen data.
\Cref{ss:mean_covariance_comparison} demonstrates the ROM accuracy in terms of mean and covariance errors in the weak sense. For the covariance error, we also discuss how the covariance evolves in time when the implicit Euler--Maruyama scheme is used to evolve the stochastic ROM learned under this framework. 
\Cref{ss:wave_fields_comparison} compares the wave fields produced by the learned stochastic ROMs with those from the experimental data.
\Cref{ss:PSD_comparison_from_ROM results} presents results in the frequency domain, by comparing the data and ROM-reconstructed power spectral densities, which are commonly used to in the context of capillary wave turbulence.
Finally, \Cref{ss:Testing_performance_of_stochastic_ROM} shows the predictive capability of the stochastic ROMs on held-out testing data by demonstrating their ability to predict the three-wave coupling phenomenon that characterizes capillary wave turbulence.

\subsection{Mean and covariance accuracy of the stochastic ROM}    
\label{ss:mean_covariance_comparison}
Using the ROM dimension selection procedure proposed in~\Cref{ss:r_selection_based_on_mean_cov_error_tradeoff}, we obtain optimal ROMs that achieve low weak errors in the mean and covariance of the reduced state samples.
\begin{figure}[h]
\centering
\begin{subfigure}[t]{0.49\linewidth}
    \centering
    \includegraphics[width=\linewidth]{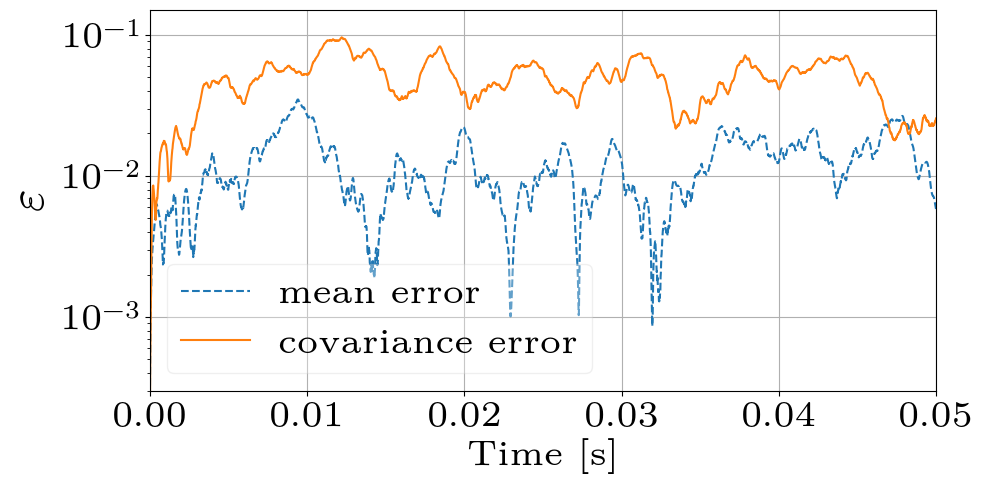}
    % \caption{The case of $\textnormal{Ca}_\textnormal{ac} = 0.83 \times 10^{-3}$.}
    \label{fig:mean_cov_error_over_time_0p10}
\end{subfigure}
% \hfill
\begin{subfigure}[t]{0.49\linewidth}
    \centering
    \includegraphics[width=\linewidth]{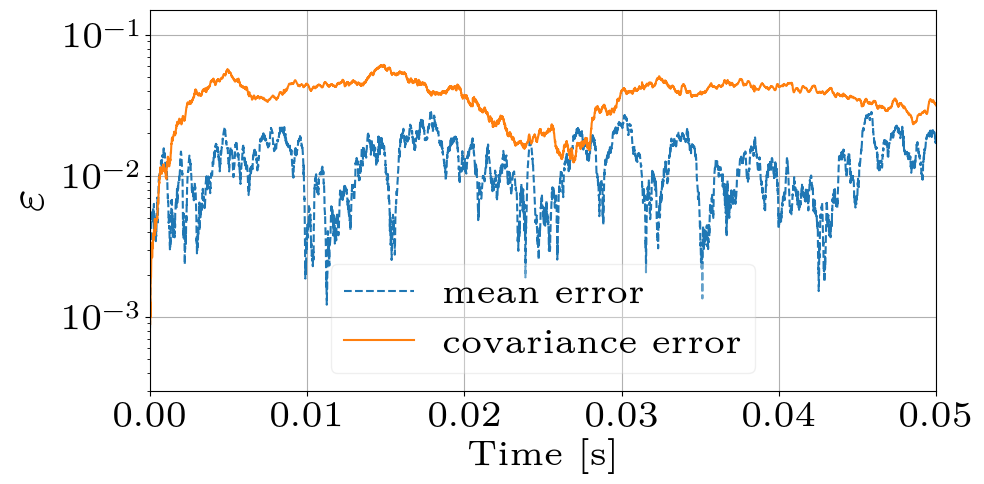}
    % \caption{The case of $\textnormal{Ca}_\textnormal{ac} = 1.76 \times 10^{-3}$.}
    \label{fig:mean_cov_error_over_time_0p20}
\end{subfigure}
\vspace{-1em}
\caption{Time evolution of the mean and covariance errors, computed for each time step in Eq.~\eqref{eq:weak_mean_covaraince_error}. The left and right plots represent the cases of $\textnormal{Ca}_\textnormal{ac} = 0.83 \times 10^{-3}$ and $\textnormal{Ca}_\textnormal{ac} = 1.76 \times 10^{-3}$, with the ROM dimensions of $r=23$ and $r=26$ chosen by the proposed criterion~\eqref{eq:r_selection_criterion}, respectively.}
\label{fig:mean_covariance_error_propagation}
\end{figure}
\Cref{fig:mean_covariance_error_propagation} shows the time evolution of the weak mean and covariance errors, computed using reduced states propagated by the learned stochastic ROMs.
Unlike the time-averaged errors defined in Eq.~\eqref{eq:weak_mean_covaraince_error}, the errors here are evaluated at each time step. The instantaneous mean and covariance errors at each time step $k$, where $k=0,\ldots,s$, are obtained using Eq.~\eqref{eq:weak_mean_covaraince_error} by evaluating the corresponding error at a single time step, without the summations over time.
% \begin{equation}
% \varepsilon_{\textnormal{mean},k} = \frac{\left \| \widehat{\bE}_{k}^L - \widehat{\bE}^{L,\textnormal{ROM}}_{k} \right \|_2^2}{\left \| \widehat{\bE}_{k}^L \right \|_2^2}, \quad \varepsilon_{\textnormal{cov},k} = \frac{\left \| \widehat{\bC}_{k}^L - \widehat{\bC}^{L,\textnormal{ROM}}_{k} \right \|_F^2}{\left \|\widehat{\bC}_{k}^L \right \|_F^2}.
% \label{eq:mean_cov_error_over_time}
% \end{equation}
%
In~\Cref{fig:mean_covariance_error_propagation}, an immediate observation is that both the mean and covariance errors remain bounded over the entire simulated time horizon rather than accumulating, indicating that the learned stochastic ROMs produce stable reconstructions.
Specifically, for the cases of $\textnormal{Ca}_\textnormal{ac} = 0.83 \times 10^{-3}$ and $\textnormal{Ca}_\textnormal{ac} = 1.76 \times 10^{-3}$, the maximum weak mean errors are $3.48 \times 10^{-2}$ and $2.87 \times 10^{-2}$, respectively, whereas the maximum weak covariance errors are $9.57 \times 10^{-2}$ and $6.14 \times 10^{-2}$, across the simulated time $t\in[0,0.05] \, \textnormal{s}$.
In addition, the time-averaged error metrics defined in Eq.~\eqref{eq:weak_mean_covaraince_error}, which are used in the ROM dimension selection criterion~\eqref{eq:r_selection_criterion}, yield $(\varepsilon_\textnormal{mean}, \varepsilon_\textnormal{cov})=(1.37 \times 10^{-2}, \, 5.19 \times 10^{-2})$ for $\textnormal{Ca}_\textnormal{ac}=0.83\times10^{-3}$, and $(\varepsilon_\textnormal{mean}, \varepsilon_\textnormal{cov}) = (1.35 \times 10^{-2}, \, 3.80 \times 10^{-2})$ for $\textnormal{Ca}_\textnormal{ac}=1.76\times10^{-3}$.
Together, these results demonstrate that the proposed ROM dimension selection criterion~\eqref{eq:r_selection_criterion} successfully identifies stochastic ROM dimensions that achieve low mean and covariance errors throughout the simulated time horizon.
For both cases in~\Cref{fig:mean_covariance_error_propagation}, the covariance error consistently exceeds the mean error by less than one order of magnitude. This shows that accurately reconstructing the covariance dynamics is more challenging than reconstructing the mean dynamics in this application. 
To better understand this behavior, we next analyze the covariance evolution under the implicit Euler--Maruyama time integration scheme.

\subsubsection{Covariance evolution over time}   \label{sss:covariance_over_time}

\begin{figure}[!ht]
\centering
\begin{subfigure}[t]{0.49\linewidth}
    \centering
    \includegraphics[width=\linewidth]{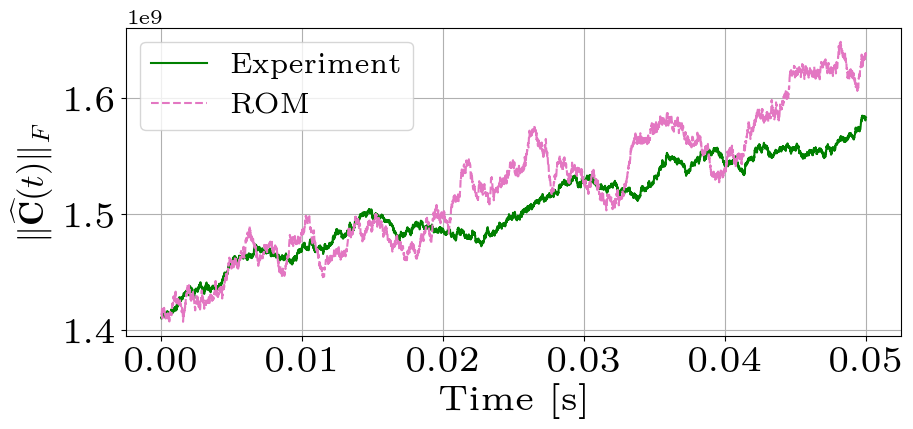}
    % \caption{}
    \label{fig:covariance_growth_0p30}
\end{subfigure}
\begin{subfigure}[t]{0.49\linewidth}
    \centering
    \includegraphics[width=\linewidth]{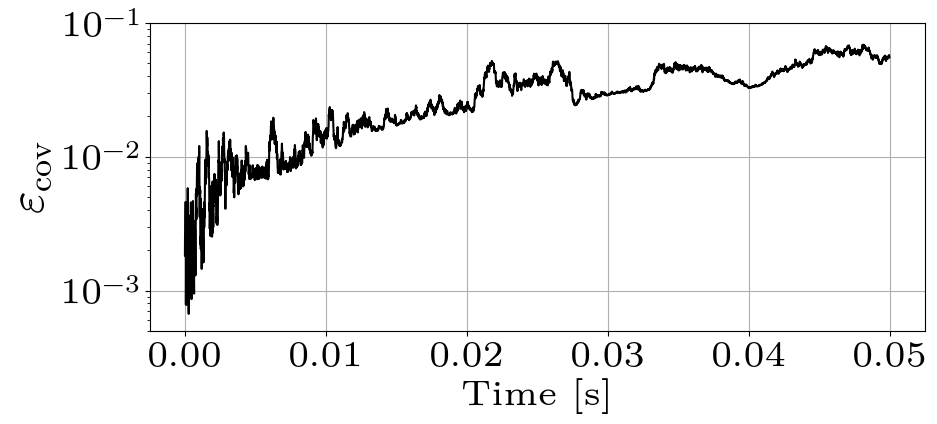}
    % \caption{}
    \label{fig:mean_cov_error_over_time_0p30}
\end{subfigure}
\vspace{-1em}
\caption{Left: Time evolution of the covariance of the reduced states, comparing between the experimental data and the ROM result. Right: Weak covariance error evolution over time. The case of $\textnormal{Ca}_\textnormal{ac} = 2.84 \times 10^{-3}$ is shown.}
\label{fig:covariance_growth}
\end{figure}
\Cref{fig:covariance_growth} compares the time evolution of the reduced-state covariance obtained from the experimental data and that generated by the stochastic ROM. 
For this case, the maximum covariance error is $7.04 \times 10^{-2}$, while the time-averaged weak covariance error defined in Eq.~\eqref{eq:weak_mean_covaraince_error} is $\varepsilon_{\textnormal{cov}} = 3.64 \times 10^{-2}$, using the ROM dimension $r=21$.
As shown in \Cref{fig:covariance_growth}, the weak covariance error grows over the simulated time $t \in [0, 0.05] \, \textnormal{s}$.
This behavior can be understood by covariance evolution induced by the implicit Euler--Maruyama time discretization of the stochastic ROM~\eqref{eq:bilinear_stochastic_ROM}, which governs how the covariance propagates from one time step to the next. 
Under this scheme, the reduced states evolve in discrete time with uniform time step~$h = t_{k+1} - t_{k}$, where $k = 0, \ldots, s-1$, through
\begin{equation}
    \widehat{\bX}_{k+1} - \widehat{\bX}_{k} = \left[ \widehat{\bA} \widehat{\bX}_{k+1} + \widehat{\bB} \bu_k + \sum_{\ell=1}^{m} \widehat{\bN}_{\ell} \widehat{\bX}_{k+1} u_{\ell,k} \right] h + \widehat{\bM} \bxi_k,
\label{eq:implicit_EM_scheme}
\end{equation}
where $u_{\ell,k} = u_\ell(t_k)$, and $\bxi_k \sim \cN(\mathbf{0},\, h\mathbf{I}_d)$, with~$\bI_d \in \real^{d \times d}$ being an identity matrix, is an independent Gaussian increment at each step. 
Equation~\eqref{eq:implicit_EM_scheme} can be rewritten as
\begin{equation}
    \widehat{\bX}_{k+1} = \widehat{\cT}_{k} \left[ \widehat{\bX}_{k} + h \widehat{\bB} \bu_k + \widehat{\bM} \sqrt{h} \bxi_k \right] ,
\label{eq:state_transition_implicit_EM}
\end{equation}
where the \textit{state transition matrix} $\widehat{\cT}_{k} \coloneqq \left[ \bI_r - h \left( \widehat{\bA} + \sum_{\ell=1}^{m} \widehat{\bN}_{\ell} u_{\ell,k} \right) \right]^{-1} \in \real^{r \times r}$.
Note that the covariance dynamics in the discrete-time domain are determined by the stochastic terms in Eq.~\eqref{eq:state_transition_implicit_EM}; the deterministic input term~$h \widehat{\bB} \bu_k$ does not contribute to the covariance of~$\widehat{\bX}_{t}$.
Defining the mean in the discrete time step $t_k$ as $\widehat{\bE}_k \coloneqq \mathbb{E} [\widehat{\bX}(t_k)]$, the covariance is estimated as
$\widehat{\bC}_{k} = \mathbb{E} \left[ \left( \widehat{\bX}_{k} - \widehat{\bE}_k \right) \left( \widehat{\bX}_{k} - \widehat{\bE}_k \right)^\top \right]$.
The covariance at the next time step $t_{k+1}$ can then be written as
\begin{equation}
    \widehat{\bC}_{k+1} = \mathbb{E} \left[ \left( \widehat{\bX}_{k+1} - \widehat{\bE}_{k+1} \right) \left( \widehat{\bX}_{k+1} - \widehat{\bE}_{k+1} \right)^\top \right].
\label{eq:covariance_discrete}
\end{equation}
To understand the covariance propagation, we begin by analyzing the first term on the right-hand side of~Eq.~\eqref{eq:state_transition_implicit_EM}. To do so, we plug this term into~Eq.~\eqref{eq:covariance_discrete}, which yields
\begin{equation}
\begin{split}
    \widehat{\bC}_{k+1}
    &= \mathbb{E} \left[ \left( \widehat{\cT}_{k} \widehat{\bX}_{k} - \widehat{\cT}_{k} \widehat{\bE}_k \right) \left( \widehat{\cT}_{k} \widehat{\bX}_{k} - \widehat{\cT}_{k} \widehat{\bE}_k \right)^\top \right] \\
    &= \widehat{\cT}_{k} \cdot \mathbb{E} \left[ \left( \widehat{\bX}_{k} - \widehat{\bE}_k \right) \left( \widehat{\bX}_{k} - \widehat{\bE}_k \right)^\top \right] \cdot \widehat{\cT}_{k}^\top \\
    &= \widehat{\cT}_{k} \widehat{\bC}_{k} \widehat{\cT}_{k}^\top.
\end{split}
\label{eq:covariance_discrete_from_drift}
\end{equation}
This resulting term describes how the covariance of the stochastic ROM state~$\widehat{\bX}_{t}$ evolves in discrete time caused by the drift dynamics encapsulated in $\widehat{\cT}_{k}$.
In addition, we do the same analysis using the diffusion term in Eq.~\eqref{eq:state_transition_implicit_EM}. Plugging it into Eq.~\eqref{eq:covariance_discrete} results in
\begin{equation}
\begin{split}
    \widehat{\bC}_{k+1}
    &= \mathbb{E} \left[ \left( \widehat{\cT}_{k} \widehat{\bM} \sqrt{h} \bxi_k \right) \left( \widehat{\cT}_{k} \widehat{\bM} \sqrt{h} \bxi_k \right)^\top \right] \\
    &= h \widehat{\cT}_{k} \widehat{\bM} \mathbb{E} \left[ \bxi_k \bxi_k^\top \right] \widehat{\bM}^\top \widehat{\cT}_{k}^\top \\
    &= h \widehat{\cT}_{k} \widehat{\bM} \widehat{\bM}^\top \widehat{\cT}_{k}^\top.
\end{split}
\label{eq:covariance_discrete_from_diffusion}
\end{equation}
Collecting the results in~Eqs.~\eqref{eq:covariance_discrete_from_drift} and~\eqref{eq:covariance_discrete_from_diffusion}, we obtain the following covariance difference equation when the implicit Euler--Maruyama scheme is used for the stochastic ROM simulation:
\begin{equation}
    \widehat{\bC}_{k+1} = \widehat{\cT}_{k} \widehat{\bC}_{k} \widehat{\cT}_{k}^\top + h \widehat{\cT}_{k} \widehat{\bM} \widehat{\bM}^\top \widehat{\cT}_{k}^\top.
\label{eq:covairance_evolution_implicit_EM_scheme}
\end{equation}
%
%
%%%%%
\subsubsection{Interpretation of the covariance difference equation}   \label{sss:interpretation_of_covariance_difference_equation}
Equation~\eqref{eq:covairance_evolution_implicit_EM_scheme} reveals that the covariance evolution consists of two distinct mechanisms: propagation of existing covariance through the learned drift dynamics and generation of additional covariance through the learned diffusion dynamics.
In particular, the first term in the right-hand side of Eq.~\eqref{eq:covairance_evolution_implicit_EM_scheme} describes how the covariance at time $k$ is propagated by the learned drift dynamics through the state transition matrix~$\widehat{\cT}_k$.
The second term, on the other hand, represents the added covariance associated with the stochastic forcing, where the learned diffusion coefficient matrix~$\widehat{\bM}$ determines the scaling and structure of the additional covariance.
Notably, this diffusion operator is also multiplied by $\widehat{\cT}_{k}$ and $\widehat{\cT}_{k}^\top$, which means that the drift and diffusion are intrinsically coupled in the covariance dynamics.
Thus, the learned drift dynamics do not merely propagate the existing covariance, but also influence the newly added one.
This term explains that the covariance evolution depends on both the drift and diffusion operators simultaneously, unlike the mean evolution, which depends only on the drift. 
\par
Since the state transition matrix~$\widehat{\cT}_k$ governs both covariance propagation and diffusion-induced covariance addition, its eigenvalues provide insight into the behavior of the covariance dynamics at each time step.
To assess this, we compute the eigenvalues of the state transition matrix at all time steps. 
For the case considered herein, where the ROM dimension is selected as $r=21$, the real parts of all 21 eigenvalues remain tightly clustered around $1$ throughout all time steps. Specifically, they remain within a narrow interval approximately between $0.99999$ and $1.00001$, corresponding to deviations on the order of $10^{-5}$.
Moreover, the maximum real part remains nearly constant over time at approximately $1.0000125$. Although this small deviation from 1 indicates a slight growth tendency in the covariance propagation, no instability is observed over the simulated interval $t\in[0,0.05]\,\textnormal{s}$ with the time step of approximately $8.68\times10^{-6}\,\textnormal{s}$.
In fact, as shown in~\Cref{fig:covariance_growth}, both the experimental and ROM covariance exhibit a growing tendency over the entire time horizon, for this specific case. This indicates that the learned stochastic ROM well captures the increasing variability observed in the experimental reduced-state dynamics.
This also implies that, although the weak covariance error in~\Cref{fig:covariance_growth} gradually increases over time, this growth occurs while the ROM covariance follows the same overall trend as the experimental covariance.
This suggests that the error is primarily associated with small discrepancies in the covariance evolution rather than an unstable covariance growth. 
\par
Overall, these analyses provide a diagnostic perspective for interpreting covariance errors in stochastic OpInf ROMs. Specifically, when larger covariance errors occur, examining the spectral properties of the state transition matrix~$\widehat{\cT}_k$ at each time step $k$, which governs covariance propagation through the implicit Euler--Maruyama scheme, can help identify whether the learned drift dynamics contribute to inaccurate covariance evolution.
Lastly, the covariance error can also originate from inaccurate learning of the diffusion operator~$\widehat{\bM}$, included in the second term of Eq.~\eqref{eq:covairance_evolution_implicit_EM_scheme}.
Since the diffusion operator directly determines the covariance contribution generated by the stochastic forcing, inaccurate learning of~$\widehat{\bM}$ can lead to discrepancies in the covariance evolution even when the drift operators are accurately learned. 
Taken together, these observations support our discussions under~\Cref{fig:mean_covariance_error_propagation} as to why capturing the covariance is more challenging than capturing the mean alone.

\subsection{Comparison of experimental and reconstructed ROM wave fields}
\label{ss:wave_fields_comparison}
\begin{figure}[!ht]
    \centering
    % Column headers
    \begin{tabular}{p{0.33\linewidth}p{0.25\linewidth}p{0.35\linewidth}}
    \centering\textbf{Experiment} &
    \centering\textbf{ROM} &
    \centering\textbf{Error}
    \end{tabular}
    \begin{subfigure}[t]{1.00\linewidth}
        \centering
        \includegraphics[width=1.00\linewidth]{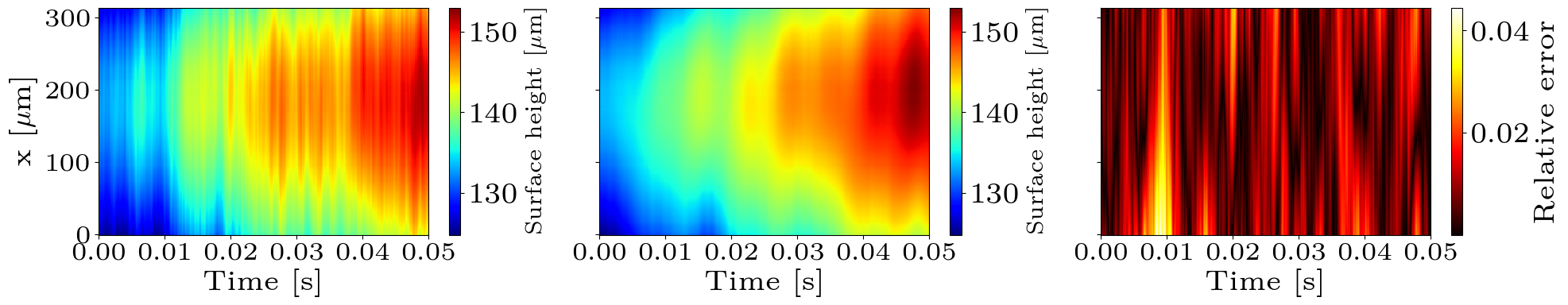}
        % \caption{The case of $\textnormal{Ca}_\textnormal{ac} = 0.83 \times 10^{-3}$.}
        \label{fig:wave_field_comparison_0p10}
    \end{subfigure}
    \vfill
    \begin{subfigure}[t]{1.00\linewidth}
        \centering
        \includegraphics[width=\linewidth]{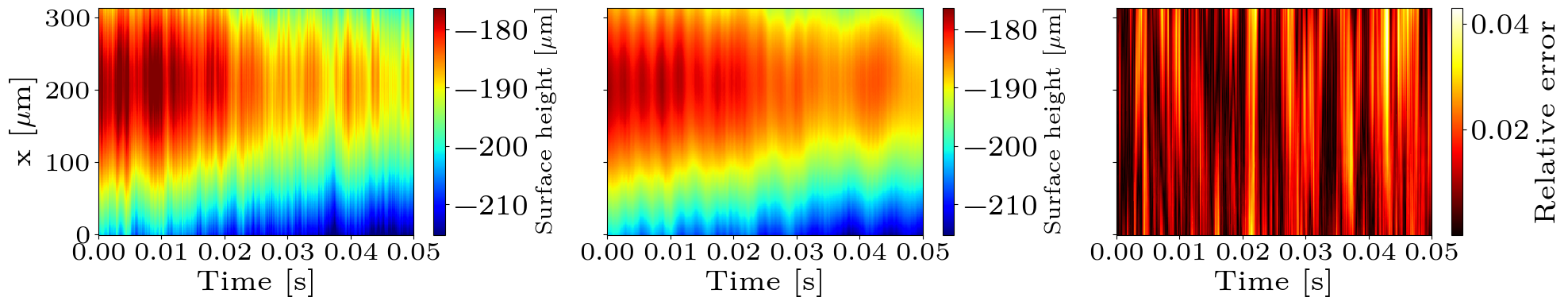}
        % \caption{The case of $\textnormal{Ca}_\textnormal{ac} = 1.15 \times 10^{-3}$.}
        \label{fig:wave_field_comparison_0p15}
    \end{subfigure}
    \vfill
    \begin{subfigure}[t]{1.00\linewidth}
        \centering
        \includegraphics[width=\linewidth]{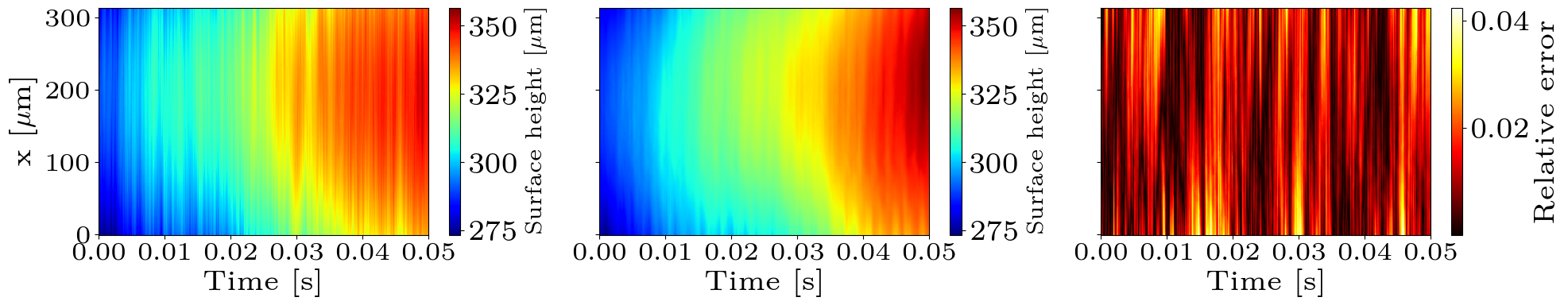}
        % \caption{The case of $\textnormal{Ca}_\textnormal{ac} = 1.76 \times 10^{-3}$.}
        \label{fig:wave_field_comparison_0p20}
    \end{subfigure}
    \vfill
    \begin{subfigure}[t]{1.00\linewidth}
        \centering
        \includegraphics[width=\linewidth]{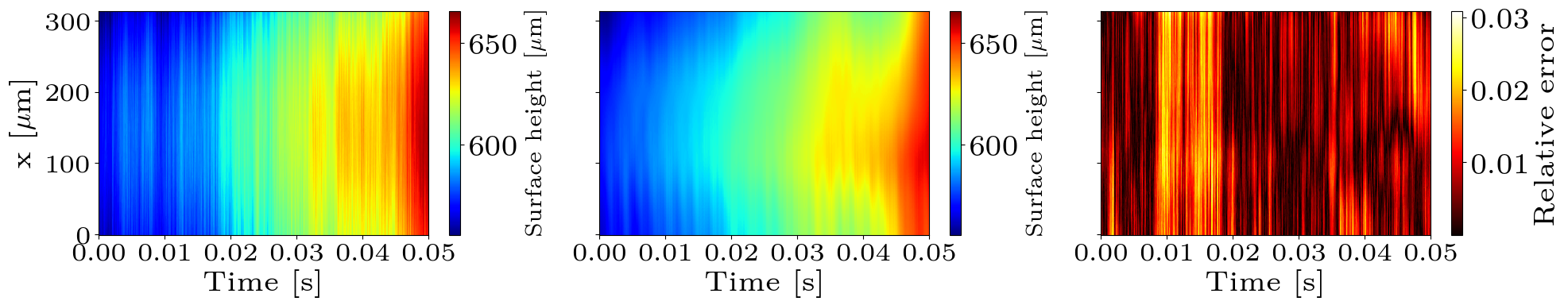}
        % \caption{The case of $\textnormal{Ca}_\textnormal{ac} = 1.99 \times 10^{-3}$.}
        \label{fig:wave_field_comparison_0p25}
    \end{subfigure}
    \vspace{-1em}
    \caption{Comparison of the experimental and ROM-generated wave fields. The pointwise relative error, defined in Eq.~\eqref{eq:pointwise_relative_error_capillary_wave}, is used in the right column. From top to bottom, the rows represent the cases of $\textnormal{Ca}_\textnormal{ac} = 0.83\times10^{-3}$, $1.15\times10^{-3}$, $1.76\times10^{-3}$, and $1.99\times10^{-3}$.}
    \label{fig:capillary_wave_field_comparison}
\end{figure}
After simulating the learned stochastic ROM simulated in the reduced space using the initial conditions of the $L$ reduced state samples from the experiment (see~\Cref{tab:experimental_datasets_capillary_wave} for $L$ of each case), each resulting trajectory is lifted back to the original high-dimensional space $\real^n$ via the POD basis used for dimensionality reduction.
As we explained in~\Cref{sss:stochastic_ROM_error_metrics}, our goal is to obtain accurate ROMs in the weak sense; getting strong approximations for the individual stochastic processes is infeasible due to the inaccessibility to the process noise in the DHM system. Thus, the wave-field comparison is done in the weak sense by comparing ensemble-averaged wave fields rather than the individual sample's field.
Similar to the empirical sample mean of the reduced states defined in Eq.~\eqref{eq:empirical_mean_reduced_state}, the mean wave field is computed by averaging over the $L$ sample trajectories as
$$
\bE^{L, (\cdot)}_k = \frac{1}{L} \sum_{j=1}^{L}\bX^{(\cdot)}(t_k, \omega_j) \in\mathbb{R}^{n},
$$
where $\bX^{(\cdot)}(t_k,\omega_j)$ denotes the $n$-dimensional experimental or reconstructed ROM wave field ($(\cdot) \in \{\textnormal{exp},\textnormal{ROM} \}$) at the $k$th time step and the $j$th sample trajectory.
The mean wave surface height at the $i$th spatial location and $k$th time step is denoted as
$
\bE_{i,k}^{L,(\cdot)} = \left[ \bE^{L,(\cdot)}_k \right]_{i},
$
where $i=1,\ldots,200$.
The pointwise relative error between the experimental and ROM mean wave fields is then defined as
\begin{equation}
\Upsilon_{i,k}^{\mathrm{rel}} \coloneqq \frac{\left|\bE_{i,k}^{L,\mathrm{exp}}-\bE_{i,k}^{L,\mathrm{ROM}}\right|}{\left|\bE_{i,k}^{L,\mathrm{exp}}\right|}.
\label{eq:pointwise_relative_error_capillary_wave}
\end{equation}
\Cref{fig:capillary_wave_field_comparison} compares the experimental and reconstructed mean wave fields over the one-dimensional spatial domain $x\in[0,300]\,\upmu\mathrm{m}$ and the time interval $t\in[0,0.05]\,\mathrm{s}$, for several representative cases.
For the four cases of $\textnormal{Ca}_\textnormal{ac} = 0.83\times10^{-3}$, $1.15\times10^{-3}$, $1.76\times10^{-3}$, and $1.99\times10^{-3}$, the selected ROM dimensions are $r=23$, $19$, $26$, and $5$, respectively.
Overall, the stochastic ROM accurately captures the spatiotemporal evolution of the experimental wave fields across all cases presented.
The reconstructed wave patterns closely match the experimental data, while the pointwise relative error remains low throughout the spatial and temporal domains.
In particular, the maximum relative error is below $0.04$ for all cases, demonstrating the ability of our ROMs to faithfully capture the complex dynamics of capillary wave phenomena.

\subsection{Comparison of power spectral densities}
\label{ss:PSD_comparison_from_ROM results}
We now also evaluate the fidelity of the ROM reconstructions in the frequency domain.
To do so, we compute power spectral densities (PSDs) via Welch's method, for both the experimental data and the corresponding reconstructions generated by the learned stochastic ROM over the same batch of realizations. We compare them to assess how well the the ROMs capture the spectral energy distribution of the experimental data.
\begin{figure}[!ht]
\centering
\begin{subfigure}[t]{0.49\linewidth}
    \centering
    \includegraphics[width=\linewidth]{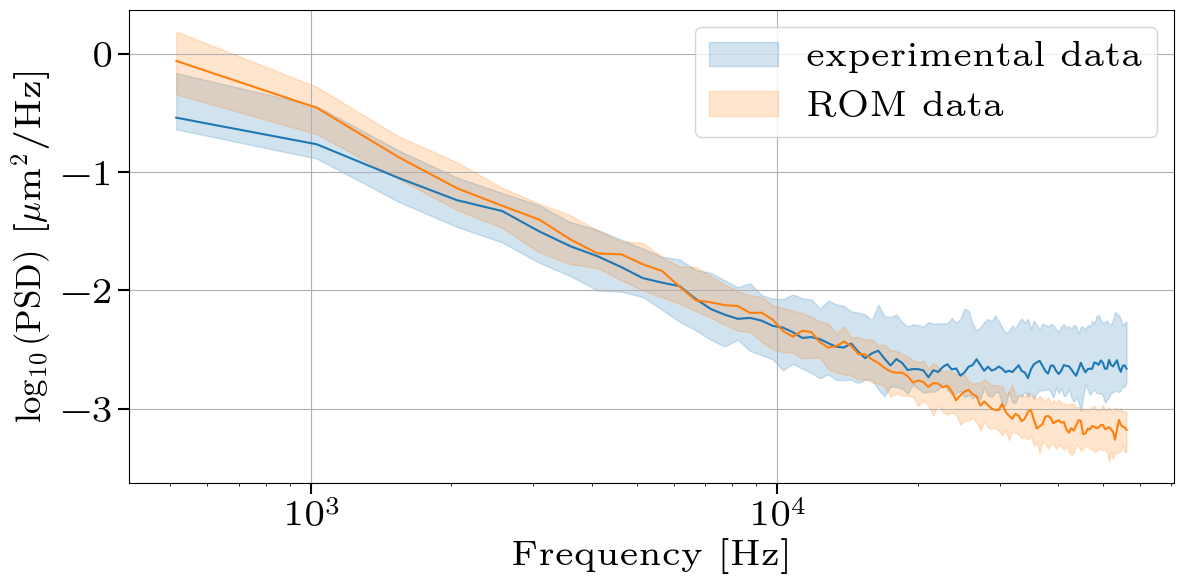}
    % \caption{The case of $\textnormal{Ca}_\textnormal{ac} = 2.84 \times 10^{-3}$.}
    \label{fig:PSD_comparison_0p30}
\end{subfigure}
% \hfill
\begin{subfigure}[t]{0.49\linewidth}
    \centering
    \includegraphics[width=\linewidth]{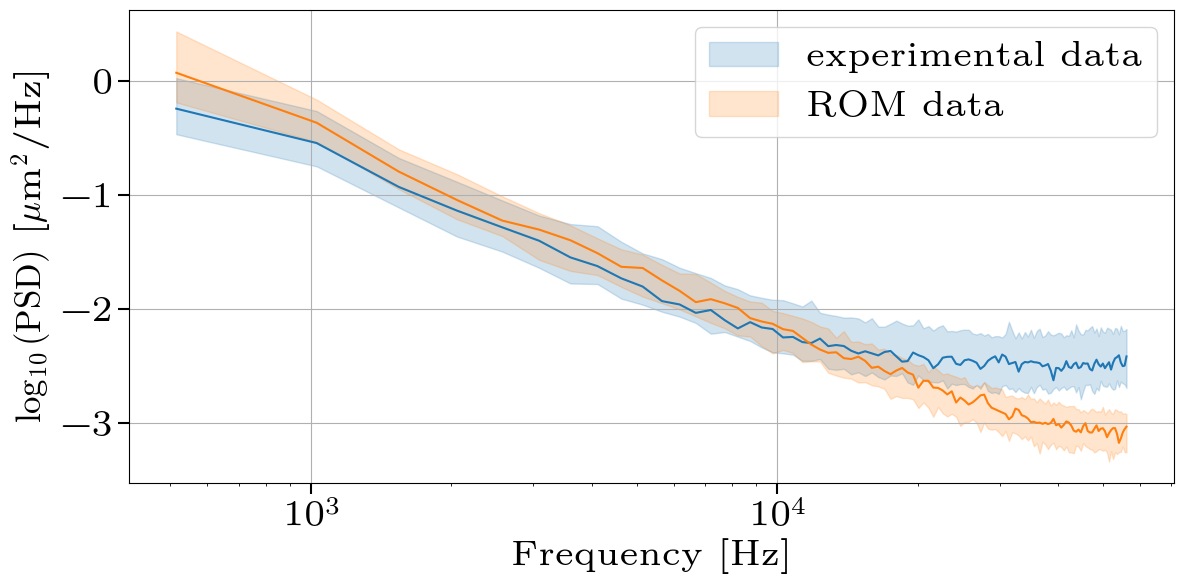}
    % \caption{The case of $\textnormal{Ca}_\textnormal{ac} = 3.32 \times 10^{-3}$.}
    \label{fig:PSD_comparison_0p35}
\end{subfigure}
\vspace{-1em}
\caption{Comparison of the power spectral densities (PSDs) computed from the time series at the midpoint of the 1D spatial domain. The left and right plots represent the cases of $\textnormal{Ca}_\textnormal{ac} = 2.84 \times 10^{-3}$ and $\textnormal{Ca}_\textnormal{ac} = 3.32 \times 10^{-3}$, respectively. For each case, PSDs are computed from all $L=272$ realizations; the shaded band shows their spread, and the solid curve shows their mean.}
\label{fig:PSD_comparison}
\end{figure}
\Cref{fig:PSD_comparison} shows the comparison of PSDs between the experimental data and the reconstructed ROM states. 
For both cases shown, the PSDs obtained from $L=272$ ROM trajectories match reasonably well with those from the experimental data, particularly for $f < 10^{4} \, [\textnormal{Hz}]$.
This indicates that the ROM reproduces the overall spectral content and energy levels at low and intermediate frequencies.
This partially reflects the regularization strategy discussed in~\Cref{ss:regularization_hyperparameters_selection_based_on_spectral_contents}; without it, spurious aliased content and an inaccurate amount of stochastic energy would otherwise distort the ROM-produced spectra.
At higher frequencies ($f > 10^{4} \, [\textnormal{Hz}]$), however, the ROM-reconstructed PSDs underestimate the experimental spectra by approximately half an order of magnitude, for both cases.
This may be due to the inherently low-dimensional nature of the ROM. The selected ROM dimensions are $r=15$ and $r=13$ for the cases $\textnormal{Ca}_\textnormal{ac} = 2.84 \times 10^{-3}$ and $\textnormal{Ca}_\textnormal{ac} = 3.32 \times 10^{-3}$, respectively, which may be insufficient to resolve high-frequency spectral content.
Since the POD basis is truncated to retain the most energetic structures in the data, the lower-energy, finer-scale portions of the dynamics that are associated with high-frequency capillary wave surface fluctuations are captured less accurately. 
Nonetheless, the spectral energy decay up to that range is reasonably captured, indicating that the ROM reproduces the dominant scales of the capillary wave dynamics well.

\subsection{Predictive capability for physical fidelity of the stochastic OpInf ROM}    
\label{ss:Testing_performance_of_stochastic_ROM}

\subsubsection{Bispectrum and bicoherence}
\label{sss:bispectrum_bicoherence}
Capillary wave turbulence exhibits nonlinear three-wave resonant interactions, and the bispectrum and bicoherence quantify the strength and phase coherence of these couplings~\cite{kim1979digital}; see~\cite{orosco2023identification, pan2017understanding} for the use of these quantities for capillary wave turbulence.
We consider these quantities as indicators for how well the learned ROMs capture the salient physics captured in the data.
To this end, we first compute a wavelet transform of the time-dependent surface displacement data of a single spatial point (the midpoint in our case), $X(t) \mapsto \widetilde{X}(t,f)$, where $X(t)$ is the surface displacement time series and $\widetilde{X}(t,f)$ is its wavelet representation in time-frequency space.
The bicoherence is then defined as
\begin{equation}
    b_{1,2} = \frac{\left| \left< \widetilde{X}(t, f_1) \cdot \widetilde{X}(t, f_2) \cdot \widetilde{X}^*(t, f_1 + f_2) \right>_t  \right|}{\left< \left| \widetilde{X}(t,f_1) \cdot \widetilde{X}(t,f_2) \cdot \widetilde{X}^*(t,f_1+f_2) \right| \right>_t},
\label{eq:bicoherence}
\end{equation}
with $<\cdot>_t$ denoting the time average, $|\cdot|$ the modulus, and $*$ the complex conjugate. The numerator is the modulus of the bispectrum, which retains the amplitude of the phase-averaged triple product.
The bicoherence~\eqref{eq:bicoherence} quantifies the strength of quadratic (three-wave) phase coupling among frequencies $f_1$, $f_2$, and $f_3=f_1+f_2$, measuring whether the interaction between $f_1$ and $f_2$ generates energy at their sum frequency.
It is normalized to the range $[0,1]$, where $b_{1,2}=0$ and $b_{1,2}=1$ indicate zero and perfect phase coupling, respectively.
\subsubsection{ROM testing setup}
\label{sss:testing_setup}
To assess the predictive capability of our stochastic OpInf ROM, we evaluate its performance on unseen experimental data samples collected under the same experimental setup (i.e., $\mathrm{Ca}_\textnormal{ac}$) used for training.
This corresponds to a scenario where additional experiments are conducted under identical conditions. Due to the inherent stochastic variability of the capillary wave dynamics, these new experiments yield different stochastic realizations of the wave surface height trajectories than those used for the ROM training.
To this end, we randomly partition the full set of samples defined in Eq.~\eqref{eq:sample_size} into training and testing subsets of equal size, so the training to testing ratio is 1:1.
Specifically, the total number of $L$ sample trajectories is randomly divided into $L_{\textnormal{train}}$ training trajectories and $L_{\textnormal{test}}$ testing samples (e.g., $L=272$ yields $L_\textnormal{train} = L_\textnormal{test} = 136$).
We train the stochastic OpInf ROM with the training samples, using the proposed ROM dimension selection method. 
To simulate the learned ROM, we project the testing data onto the POD basis obtained from the training data. The projected initial reduced state samples are then used as initial conditions for the ROM tests.
After the ROM simulation, the resulting reduced state predictions are lifted to the high-dimensional space, where we compute the bispectrum and bicoherence.
\begin{figure}[!ht]
\centering
\begin{subfigure}[t]{1.00\linewidth}
    \centering
    \includegraphics[width=\linewidth]{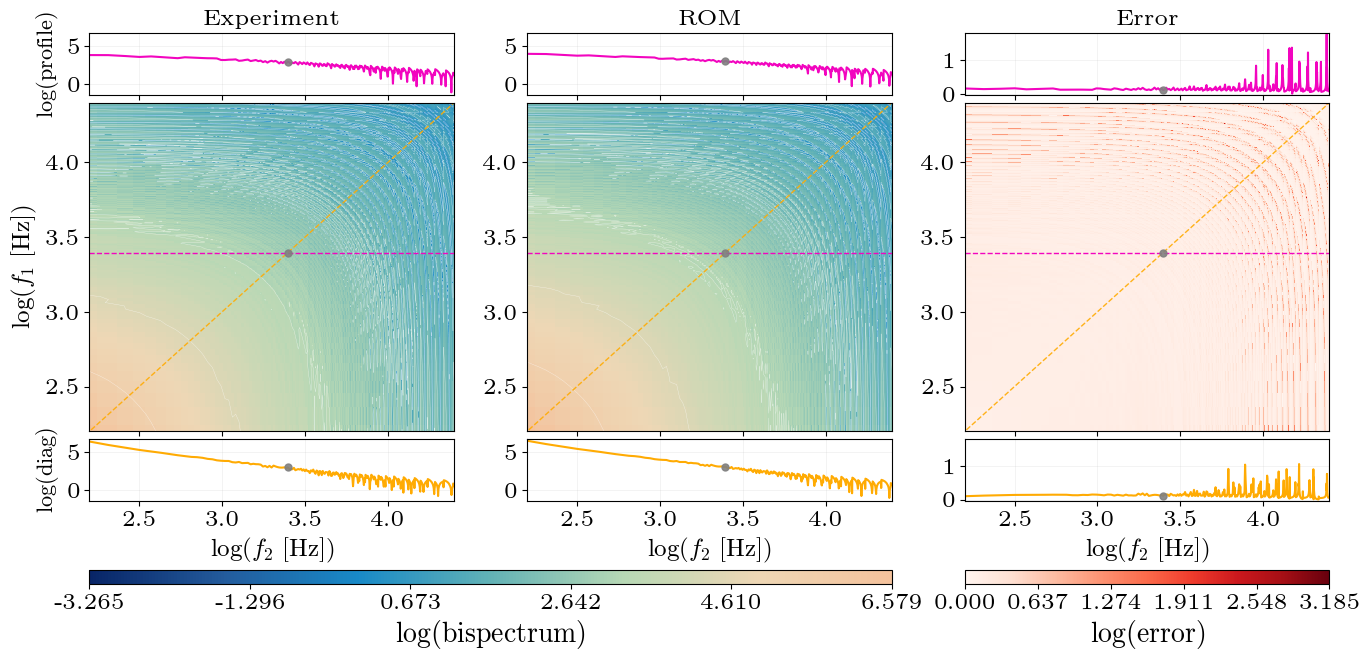}
    % \caption{Bispectrum}
    \label{fig:bispectrum_comparison_testing}
\end{subfigure}
% \hfill
\begin{subfigure}[t]{1.00\linewidth}
    \centering
    \includegraphics[width=\linewidth]{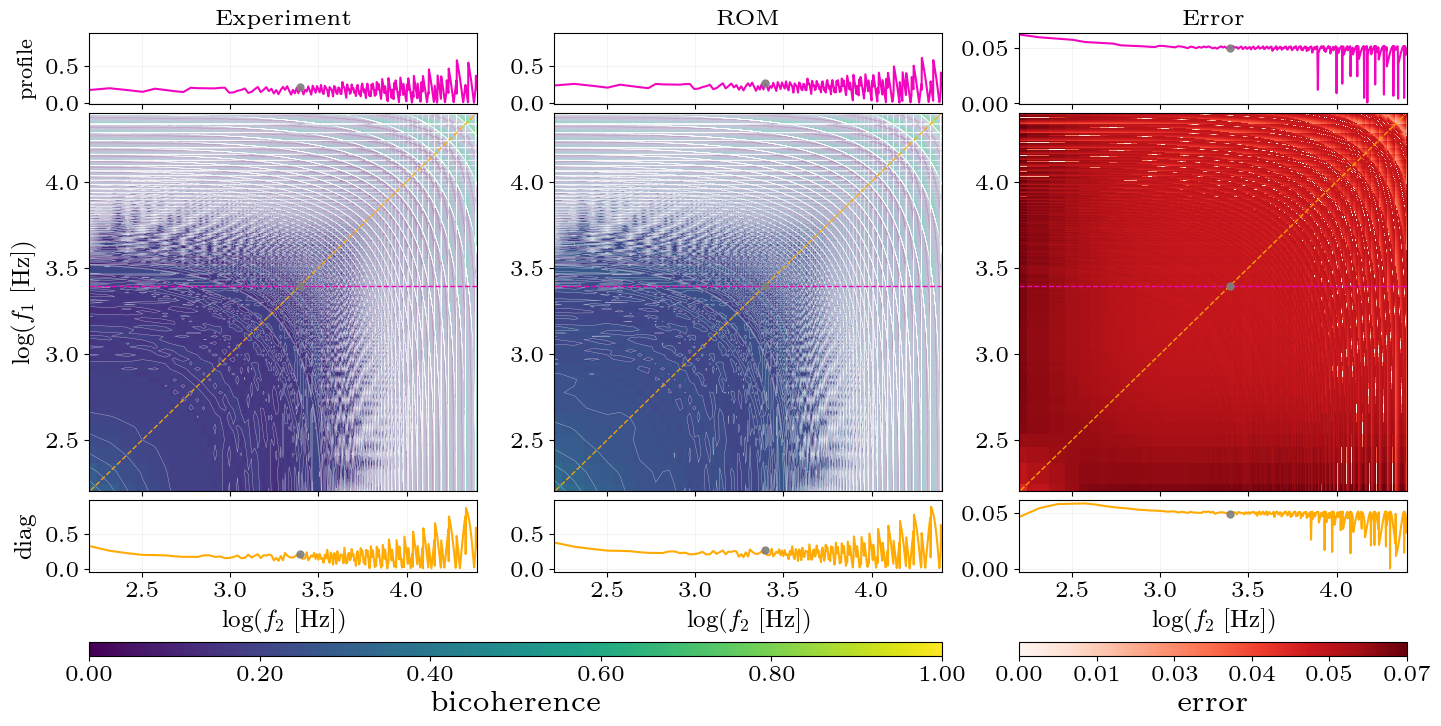}
    % \caption{Bicoherence}
    \label{fig:biscoherence_comparison_testing}
\end{subfigure}
\vspace{-1em}
\caption{ROM testing results on unseen datasets: bispectrum (top) and bicoherence (bottom). For each quantity, the upper profile shows a horizontal profile of $\log(f_1 \, [\textnormal{Hz}])=3.40$, which is arbitrarily extracted for the off-diagonal frequency combination ($f_1 \neq f_2$). The lower profile for both quantities represents the diagonal profile of the contour plots, corresponding to $f_1 = f_2$. The pointwise absolute errors are shown for each quantity.}
\label{fig:bispectrum_bicoherence_testing_performance}
\end{figure}
\par
\subsubsection{ROM testing results}
\label{sss:testing_results}
\Cref{fig:bispectrum_bicoherence_testing_performance} compares the bispectrum and bicoherence, computed using the held-out testing data from the experiments and the ROM-predicted results emanating from those held-out test initial conditions. 
The case of $\textnormal{Ca}_\textnormal{ac} = 0.25 \times 10^{-3}$ is considered.
For both the experimental and ROM results, we first compute the mean trajectory across the $L_\textnormal{test}=136$ testing realizations. The resulting mean time series are then used to compute the bicoherence via Eq.~\eqref{eq:bicoherence}, whose numerator is the bispectrum.
For the bispectrum in~\Cref{fig:bispectrum_bicoherence_testing_performance}, the ROM predicts the overall three-wave coupling strength well across the entire 2D frequency domain. Especially, along the representative off-diagonal slice at fixed $\log(f_1 \, [\textnormal{Hz}]) = 3.4$, the absolute error remains small up to approximately $f_2 = 10^4 \, [\textnormal{Hz}]$.
The diagonal of the bispectrum, where $f_1 = f_2 = f$, captures quadratic self-coupling of a wave at frequency $f$ with itself, producing a component at the second harmonic $2f$ through the triad interaction $f + f = 2f$.
Therefore, a large bispectrum value along the diagonal indicates strong phase-coupled second harmonic generation, the signature of nonlinear wave interactions.
The bispectrum error along the diagonal similarly remains small at low frequencies, but begins to increase around $f=5{,}000 \, [\textnormal{Hz}]$. 
In contrast, the bicoherence error remains below approximately $0.07$ (on a scale from $0$ to $1$) for both the off-diagonal and diagonal profiles.
This result highlights that our stochastic ROM accurately predicts the phase-coupling structure of the nonlinear wave interactions for capillary wave turbulence, well beyond the training data.

% \subsection{Computational time}
% \label{ss:stochastic_OpInf_computational_time}
% \begin{itemize}
%     \item In this section, we briefly provide the computational time in this stochastic ROM pipeline. 

%     \item 

% \end{itemize}

%%%%%%%%%%%%%%%%%%%%%%%%%%%% Conclusion %%%%%%%%%%%%%%%%%%%%%%%%%%%%%%%%
\section{Conclusions} \label{sec:conclusion}
In this work, we extended a stochastic OpInf ROM framework by providing its first demonstration through real experimental data.
With several improvements proposed in this work, we showed that stochastic OpInf~\cite{freitag2025learning} can learn low-dimensional SDE representations of complex microscale capillary wave turbulence directly from ultra-high-speed digital holographic microscopy measurements, enabling reduced-order modeling of experimental systems where governing equations and first-principles numerical discretizations are not readily available. 
In particular, we proposed a new ROM dimension selection strategy for stochastic OpInf that provides an alternative to conventional energy-based criteria. By explicitly accounting for both weak mean and covariance errors, the proposed approach identifies ROM dimensions that accurately capture both the dominant dynamics and the stochastic variability of the ROM-state distribution.
We further added Tikhonov regularization to the numerical implementation of operator learning with the goal to improve numerical robustness, and then analyzed it to show that it has a physically meaningful role in stochastic OpInf by influencing the spectral content of the learned stochastic dynamics.
We also investigated the covariance evolution of stochastic ROMs under the implicit Euler--Maruyama time integration scheme. This analysis revealed that covariance dynamics are governed by the coupled effects of the learned drift and diffusion operators, providing a diagnostic perspective for interpreting covariance errors and assessing stochastic ROM behavior.
The resulting stochastic OpInf ROMs accurately captured the statistical, spectral, and physical characteristics of the experimental capillary wave turbulence system across multiple excitation conditions.
Specifically, the ROMs achieved low weak mean and covariance errors, accurately reproduced the PSDs, and successfully predicted nonlinear three-wave coupling in previously unseen experimental realizations---demonstrated via bispectral and bicoherence analyses.
\par
Future work could explore more expressive parameterizations of the drift and diffusion operators, including machine learning-based approaches, to further improve the modeling capability of stochastic ROMs for strongly nonlinear and random dynamics.
Furthermore, using the full two-dimensional capillary wave field for a machine learning approach, rather than one-dimensional profiles used in this study, would provide a more comprehensive representation of the dynamics.
Also, parametric extensions of stochastic ROMs could be investigated to capture variations in the system dynamics across different parameter values and operating conditions, such as the forcing amplitude in our case.
Finally, while the proposed framework uses Wiener processes as stochastic forcing, the reduced-state distributions obtained from the experimental measurements, particularly the initial states used for ROM simulations, exhibited deviations from Gaussianity. 
Although we investigated a temporally correlated stochastic forcing based on different versions of fractal Gaussian noise processes, it did not yield noticeable performance improvements.
Future work may therefore investigate other classes of non-Gaussian and colored stochastic processes, or more general stochastic representations, to better capture the statistical characteristics of experimental systems. 

%%%%%%%%%%%%%%%%%%%%%%%% Contribution statement %%%%%%%%%%%%%%%%%%%%%%%%%%%%%
\section*{CRediT authorship contribution statement}
\textbf{Hyeonghun Kim}: Writing – Original Draft, Writing – Review \& Editing, Conceptualization, Methodology, Software, Validation, Formal analysis, Investigation, Data curation, Visualization;
\textbf{Lei Zhang}: Writing - Review \& Editing, Software, Data curation, Investigation;
\textbf{James Friend}: Writing: Review \& Editing, Investigation, Resources, Supervision, Project administration, Funding acquisition;
\textbf{Boris Kramer}: Writing – Review \& Editing, Conceptualization, Methodology, Validation, Formal analysis, Supervision, Project administration, Funding acquisition.

%%%%%%%%%%%%%%%%%%%%%%%% code %%%%%%%%%%%%%%%%%%%%%%%%%%%%%
\section*{Data availability}
The public repository \href{https://github.com/hyeonghun1/CapillaryWaveTurbulence}{https://github.com/hyeonghun1/CapillaryWaveTurbulence} contains a collection of Jupyter notebooks and Python scripts, implemented in Python 3.10, containing the code and data used in this study.
The complete experimental dataset is publicly available through~\cite{capillary_wave_datasets}. 

%%%%%%%%%%%%%%%%%%%%%%%% Acknowledgement %%%%%%%%%%%%%%%%%%%%%%%%%%%%%
\section*{Acknowledgments}
This research was financially supported by the Defense Advanced Research Projects Agency (DARPA) Defense Sciences Office (DSO) through award \#HR0011-25-3-0132 for project STORMLAP: Strong Turbulence and Rogue-Wave Modeling using Machine-Learning Assisted Predictions.

%%%%%%%%%%%%%%%%%%%%%%%% Tool disclosure %%%%%%%%%%%%%%%%%%%%%%%%%%%%%
\section*{Tool and computational resource disclosure}
During the preparation of this work, H.K. used ChatGPT and Claude Sonnet 5 to assist with improving the clarity of existing sentences, identifying potentially relevant references, and troubleshooting programming-related issues. L.Z. used ChatGPT to assist with troubleshooting programming-related issues.
All suggestions generated by these tools, including references, were independently verified by H.K. and L.Z. The author(s) reviewed and edited the content as needed and take(s) full responsibility for the content of the published article.

%%%%%%%%%%%%%%%%%%%%%%%%%%%% Appendix %%%%%%%%%%%%%%%%%%%%%%%%%%%%%%%%
% \newpage
%
% \appendix
%

% \section{Comparison of deterministic and stochastic Operator Inference reduced-order models of different model forms} \label{app:Comparison_of_bilinear_and_linear}

% \section{Discrete-time covariance dynamics under the implicit Euler-Maruyama scheme} \label{app:covariance_evolution_under_implicit_EM_scheme}

% \begin{itemize}
    
% \end{itemize}

%%%%%%%%%%%%%%%%%%%%%%%%%%%%%%%%%%%%%%%%
% BIBLIOGRAPHY
%%%%%%%%%%%%%%%%%%%%%%%%%%%%%%%%%%%%%%%
% \newpage
\bibliography{references}
\bibliographystyle{abbrv}

%\bibliographystyle{alphaurl}
% \bibliographystyle{plain}   % alphabetical order
% \bibliographystyle{unsrt}

%\newpage
%\begin{appendix}
%    \input{appendix}
%\end{appendix}
%
\end{document}